\documentclass[10pt,reqno]{amsart}
\usepackage[top=30mm,bottom=30mm,inner=20mm,outer=20mm]{geometry}                
\usepackage{upgreek}
\usepackage{graphicx}
\usepackage{amssymb}
\usepackage{epstopdf}
\usepackage{slashed}

\usepackage{libertine}
\usepackage{todonotes}
\usepackage[mathscr]{euscript}
\usepackage{hyperref}
\newtheorem{thm}{Theorem}
\newtheorem{lm}[thm]{Lemma}
\newtheorem{proposition}[thm]{Proposition}

\newcommand{\supp}{\mathrm{supp}}

\newcommand{\Tr}{\mathrm{Tr}}
\newcommand{\A}{\slashed{A}}
\newcommand{\be}{\begin{equation}}
\newcommand{\ee}{\end{equation}}
\def\Aslash{\slashed{A}} 
\def\@tvsp{\mathchoice{{}\mkern-4.5mu}{{}\mkern-4.5mu}{{}\mkern-2.5mu}{}}
\def\ltrivert{\left|\@tvsp\left|\@tvsp\left|}
\def\rtrivert{\right|\@tvsp\right|\@tvsp\right|}
\newcommand\lnorm[1]{\ltrivert#1\rtrivert}
\def\lquvert{\left|\@tvsp\left|\@tvsp\left|\@tvsp\left|}
\def\rquvert{\right|\@tvsp\right|\@tvsp\right|\@tvsp\right|}

\newcommand{\ie}{{\it i.e. }}

\usepackage{scalerel,stackengine}
\stackMath
\newcommand\reallywidehat[1]{%
\savestack{\tmpbox}{\stretchto{%
  \scaleto{%
    \scalerel*[\widthof{\ensuremath{#1}}]{\kern.1pt\mathchar"0362\kern.1pt}%
    {\rule{0ex}{\textheight}}
  }{\textheight}%
}{2.4ex}}%
\stackon[-6.9pt]{#1}{\tmpbox}%
}
\title[Fermions in external electromagnetic fields]{Thermodynamical aspects of fermions\\ in external electromagnetic fields}
\author{Romeo Brunetti}
\address{Romeo Brunetti
\newline 
\indent
Dipartimento di Matematica, Universit\`a di Trento - Via Sommarive 14, I-38123 Povo (TN), Italy.}
\email{romeo.brunetti@unitn.it}

\author{Klaus Fredenhagen}
\address{Klaus Fredenhagen \newline
\indent II. Institut f\"ur Theoretische Physik, Universit\"at Hamburg, Luruper Chaussee 149, D-22761 Hamburg, Germany.}
\email{klaus.fredenhagen@desy.de}

\author{Nicola Pinamonti}
\address{Nicola Pinamonti \newline 
\indent Dipartimento di Matematica, Universit\`a di Genova - Via Dodecaneso 35, I-16146 Genova, Italy 
\newline
\indent
Istituto Nazionale di Fisica Nucleare - Sezione di Genova, Via Dodecaneso 33, I-16146 Genova, Italy.}
\email{nicola.pinamonti@unige.it}

\begin{document}
\begin{abstract}
The thermodynamics of Dirac fields under the influence of external electromagnetic fields is studied. For perturbations which act only for finite time, the influence of the perturbation can be described by an automorphism which can be unitarily implemented in the GNS representations of KMS states, a result long known for the Fock representation. For time-independent perturbations, however, the time evolution cannot be implemented in typical cases, so the standard methods of quantum statistical mechanics do not apply. Instead we show that a smooth switching on of the external potential allows a comparison of the free and the perturbed time evolution, and approach to equilibrium, a possible existence of non-equilibrium stationary states (NESS) and Araki's relative entropy can be investigated. 
As a byproduct, we find an explicit formula for the relative entropy of
gauge invariant quasi-free states.
\end{abstract}

\maketitle

\dedicatory{\emph{Dedicated to the memory of Huzihiro Araki}}

\section{Introduction}

The thermodynamical aspects of relativistic quantum field theory are still poorly understood. There exists a characterization of equilibrium states by the KMS condition \cite{HHW}, which can be strengthened to the relativistic KMS condition proposed by Bros and Buchholz \cite{BrosBuchholz-RelativisticKMS}. 
The KMS condition reduces to the Gibbs condition if the time evolution in the vacuum representation is implemented by a Hamiltonian with sufficiently fast increasing eigenvalues and no continuous spectrum. 
The KMS condition is a requirement for correlation functions of observables as functions of time, where the time dependence is induced by a 1-parameter group of automorphisms. Controlling the time evolution and thus giving equilibrium states for free field propagating on stationary spacetime is relatively straightforward, its use to describing equilibrium states of interacting theories requires, however, control of the interacting time evolution and its comparison to the free one. Thermofield theory (see e.g.\  \cite{LeBellac}) starts from this observation to construct interacting equilibrium states.    
Within perturbation theory, thermofield theory  offers a framework which
permits us to treat some of the thermodynamical aspects of free and interacting quantum fields. Thermofield theory is useful for deriving the form of some of the coefficients appearing in nonlinear transport equations that describe the evolutions of certain mean values, such as, e.g., the Boltzmann equation \cite{Calzetta-Hu}.
Another interesting application is the attempt to understand baryogenesis in the early universe \cite{ProkopecA, ProkopecB}.
We refer to \cite{Landsman, LeBellac, Kapusta} for reviews and for further details. 
Thermofield theory suffers, however, from divergences in higher loop orders \cite{Altherr}. 
A solution to the latter problem was obtained within perturbative algebraic quantum field theory (pAQFT) \cite{FL}. There it could be shown that methods from quantum statistical mechanics, originally developed for C*-algebras, could be applied also to the *-algebras occurring in perturbation theory. This allowed for a 
relativistic treatment of Bose-Einstein condensation \cite{BFP} as well as a proof of the existence of non-equilibrium steady states (NESS) \cite{DFP, DFP2}. Moreover, first applications to physical phenomena as \emph{e.g.} plasma oscillations could be obtained \cite{Galanda1}.

In quantum field theory, going beyond perturbation theory is notoriously difficult. In 4 dimensional spacetime, a model which is exactly solvable, and, at the same time, of some physical interest is the Dirac field in an external electromagnetic field. A better understanding might be useful {\it e.g.} for the study of
Tokamaks \cite{Tokamaks} or Stellarators \cite{Stellarators}.

But even in this rather simple model, the occurring interactions are so singular that the standard methods of quantum statistical mechanics cannot be directly applied. The interaction term of the Hamiltonian can usually be defined as a quadratic form but not as an operator, and the form sum with the free Hamiltonian does not define a selfadjoint operator. The commutator of the Hamiltonian with the Dirac field, however, induces a derivation of the algebra of canonical anticommutation relations (CAR-algebra) which can be integrated to an automorphism group. In the case of a purely electric potential, these automorphisms can be implemented in the vacuum representation by a strongly continuous 1-parameter group of unitaries. Its generator can be considered as the Hamiltonian for the interacting theory. But due to an infinite renormalization of the vacuum energy its domain has only trivial intersection with the domain of the free Hamiltonian \cite{FredenhagenQED1}, so also in this indirect way there is no possibility to define the interaction term as an operator.
In the case of non-vanishing spatial components of the vector potential the automorphisms cannot be implemented \cite{Ruijsenaars77, Fredenhagen, Laz}.

The non implementability manifests itself in infinities present in the perturbative expansion of the interacting time evolution. These divergences forbid to take the sum of the perturbative series.
In the paper \cite{DeckertMerkel} the authors 
circumvent the problem by using time dependent, and in general mutually inequivalent Fock space representations. But then typical thermodynamical quantities like relative entropy become ill defined.

Within pAQFT the relation between the free and the interacting time evolution is described by a suitable cocycle which can be given in terms of the relative S-matrix of the theory. 
In cases where this relative S-matrix is at disposal, one can thus find a way to compare the free and interacting time evolutions and consequently the free and interacting equilibrium states. 

In the case of a time dependent external potential, which vanishes outside of some compact region of spacetime, the situation regarding implementability is better. In this case, the S-matrix induced by the interaction exists under mild conditions on the external potential \cite{Ruijsenaars77}. We can thus follow the same strategy as in pAQFT where the perturbatively constructed S-matrices of compactly supported interactions could be used to define the interacting theory also for interactions with no restrictions on the support, with the difference that in the situation treated in this paper these S-matrices are unitary operators in a physically interesting representation and not just formal power series.

The physical interpretation of this procedure for the case of a time independent external field is the following. The conventional approach identifies the free and the interacting field at time zero. At positive times, the  time evolution describes
the instantaneous switching on of the interaction. This might heavily perturb equilibrium states of the free theory and results in the non-implementability results mentioned above.
Instead, we apply the regularized interaction picture introduced in \cite{FL} where one
uses a smooth switching on of the interaction in some time interval $[-\epsilon,0]$ and identifies the fields at time $t=-\epsilon$. An equilibrium state $\omega_0$ of the free theory at early times is then transformed to a non-equilibrium state $\omega'$ of the interacting theory at time zero whose time evolution can be investigated in terms of a strongly continuous unitary group.

In this paper we are interested in comparing equilibrium states with respect to time evolutions under the influence of an external electromagnetic potential with states arising from equilibrium states of vanishing external potential by smoothly switching on the external field.
This is done comparing and controlling the corresponding automorphisms and the cocycle which intertwines them given in terms of relative S-matrices of the theory. 
The analysis of the relative S-matrix is thus a mean for the investigation of thermodynamic properties of the Dirac system in an external potential.

The construction of the solution of the Dirac equation in an external time varying potential is well understood in the literature. The construction of states for this system is also well known and studied in the literature. See e.g. the work of \cite{FinsterMurroRoken}
where also a construction of a preferred Hadamard state for the time dependent problem is obtained making use of spectral properties of the Fermionic projectors. However, the states obtained in that way are far from being equilibrium states and construction of equilibrium states adapting those methods does not seem to be a straightforward procedure. 
We remind the reader that already in the case of a free Dirac field in Minkowski spacetime the equilibrium states cannot be constructed by means of trace class density operators on the vacuum Fock spaces. The equilibrium states characterized by the KMS condition with inverse temperature $\beta$ give in fact the origin to mutually inequivalent representations of the corresponding observable algebra.  

Returning to the analysis of the state $\omega'$ evolved with the time evolution perturbed by the switched-on electromagnetic potential, we study the question whether, at a late time when the external potential is constant in time, it converges to an equilibrium state $\omega$ of the interacting theory . 

We shall actually see that if the external potential has spatially compact support and if it is sufficiently small in a suitable sense, convergence to equilibrium can be established.

To prove return to equilibrium in the non-perturbative way, we analyze the validity of the Cook criterion, see e.g. \cite{Kato}, and in order to prove that criterion we shall study the decay in time of the time dependent perturbation sourced by space compact external potential on static solutions of the Dirac equation.
Similar methods have been already used in the literature since long \cite{Prosser}. A more recent application is 
the investigation of the decay in time of solutions of the Dirac equations on curved background 
which can be found e.g. in  \cite{Nicolas} where the scattering of linear Dirac fields by spherically symmetric black holes is studied.

We furthermore 
determine Araki's relative entropy $\mathscr{S}(\omega',\omega)$ \cite{Araki, Araki2}\footnote{In the case of CAR algebras, the relative entropy computation of excited states with respect to equilibrium states can be found in \cite{Galanda}, (see also \cite{Longo, Casini, CiolliLongoRuzzi, Hollands, Frob1}).}, which is the relative entropy of the state $\omega$ w.r.t. $\omega'$. We show that for sufficiently small external fields convergence to equilibrium holds, but not in general if the interaction leads to bound states.  There non-equilibrium stationary states (NESS) \cite{Ruelle} can occur.

We can furthermore use the obtained formula for the relative entropy to estimate the entropy production (see e.g. \cite{JP01, JP02, JP03}) in the case of evolutions which do not show return to equilibrium.

The relative entropy for states which are described by suitable traceclass density operators and which are also quasi-free 
can be given in terms of the one particle density operator. Work in this direction has already been done in the literature, see e.g. \cite{Finster,FinsterLottinerSobolev,PfeifferSpitzer,Helling,Muller}. 
However, those methods cannot be used directly for states which are not described by suitable trace class density operators on the vacuum Fock space.  

On the other hand, the formula for Araki's relative entropy we use in this paper usually holds for normal states of suitable von Neumann algebras describing the observables of the theory and it is particularly useful when it is used for states which are not quasi-equivalent to the vacuum state.

In this paper we derive a formula for Araki's relative entropy for the perturbed equilibrium states with respect to the free equilibrium state after switching on the field. We compare it with the formula in \cite{Finster}. In both cases the formulas are given in terms of the one particle density operator. But applying the latter formula to the case considered in our paper one finds infrared divergences. 
In fact, the formula derived in our paper reduces to the results of \cite{Finster} when states described by trace class density operators are considered. 

The paper is organized as follows. In the next section, after recalling some facts about the 
(non) implementability of the M\"oller operators, 
which intertwine the free Dirac field and the Dirac field with an external potential, 
 we show how to implement the interacting time evolution on the Fock space of the free vacuum state in case the external potential is switched on adiabatically. We find that the interaction Hamiltonian can be described by a suitable one-particle Hilbert space operator, and we discuss some of its properties.

In Section \ref{se:equilibrium-states} we discuss the construction of equilibrium states for the interacting time evolution. 
In Section \ref{se:NESS}
we discuss the validity of the Cook criterion which permits to obtain return to equilibrium and stability properties of the equilibrium states. 
Furthermore, we show how one obtains NESS when the interacting Dirac operator has bound states. 
Section \ref{eq:relative-entropy} contains a formula for the relative entropy among equilibrium states of different interacting hamiltonians, its use in combination with the NESS and the entropy production.

We finally derive an explicit formula for the relative entropy between gauge invariant quasi-free states which reminds on the Kullback-Leibler divergences \cite{KL}.
The formula is a generalization of a similar result for quasi free states in a simpler situation studied in \cite{Finster}.

\section{Implementability in the Fock space}
\subsection{The regularized interaction}
We review the construction of a well-defined interaction in the framework of dynamical algebras \cite{BF19}. There, the starting point is a Lagrangian $L$ and potential interactions $F$ viewed as spacetime integrals of compactly supported polynomial densities over an appropriate configuration space (in the following, polynomials may contain kinetic terms that locally perturb the metric structure implicit in $L$)\footnote{We shall identify the interactions with the densities making use of the Stora's identity a.k.a. Action Ward Identity \cite{Stora,DF}.}. Then one introduces S-matrices $S(F)$ interpreted operationally as unitary operations that perform local changes of the dynamics defined by $L$ via the local interaction $F$. These S-matrices can be implicitly defined in terms of the following algebraic relations:
\begin{description}
    \item[Causal Factorization] $S(F-G+H)=S(F)S(G)^{-1}S(H)$ if $\supp(F-G)$ is later than $\supp(H-G)$ with respect to the causal order induced by the interaction $G$.
    \item[Dynamics] $S(F)=S(F^{u}+\delta_u L)$ where $L$ is the underlying Lagrangian, $u$ denotes a compactly supported field configuration, $F^u(\bullet)\doteq F(\bullet + u)$ the shift in the interaction and $\delta_u L$ the difference between the shifted and the original Lagrangian.
    \item[Normalization] $S(c)=e^{ic}\mathbf{1}$ for the constant interaction $c\in\mathbb R$.
\end{description}
These elements generate a group, and the associated group algebra admits a maximal C*-norm and can be completed to a C*-algebra satisfying the Haag-Kastler axioms.
For the free Lagrangian and the restriction to linear interactions ({\emph{sources}) one obtains the Weyl algebra. In perturbation theory, one obtains a representation of the underlying group algebra in terms of formal power series, and in a few cases including the sine Gordon model in 2 dimensions \cite{BaRe,BaFreRe} as well as free theories with quadratic interactions \cite{BF21}  one finds representations of the dynamical algebra in a separable Hilbert space. 

For theories with fermions one needs to enrich the framework by auxiliary Grassmann parameters \cite{BDFR-Fermi,AB2}. For gauge theories, the construction is more complicated \cite{BDFR-Gauge}.

We start from a time-independent Lagrangian $L$ and consider the case of a time-independent interaction density $V$ with compact spatial support (and without a kinetic term). We study the process of switching on the interaction before $t=0$ by multiplying $V$ with a smooth function $h(t)$ with $h(t)=1$ for $t>0$ and $h(t)=0$ for $t<-\epsilon$. 

Following \cite{BF19},
the interacting observables $S_{hV}(F)$ are, according to Bogoliubov's formula, described by 
\be
S_{hV}(F)\doteq S((h-h_T)V)^{-1}S((h-h_T)V+F)
\ee
with $h_T(t)=h(t-T)$ and $T$ sufficiently large such that $\supp h_T$ is later than $\supp F$. The right-hand side does not depend on the choice of $T$ due to Causal Factorization.

The unperturbed time involution acts by an automorphism $\alpha_t$, 
\be
\alpha_t(S(F))=S(F_t)
\ee
where $F_t$ is the interaction $F$ shifted by $t$. For $t>0$ and interactions $F$ with the entire support for positive times, the perturbed time evolution $\alpha_t^V$ is
\be
\alpha_t^V(S_{hV}(F))=S_{hV}(F_t)\ .
\ee
The time evolution $\beta_t^{hV}$ in the interaction picture for $t>0$ 
\be
\alpha_t^V=\beta_t^{hV}\circ\alpha_{t}
\ee
is unitarily implemented,
\be
\begin{split}
&\alpha_t^V(S_{hV}(F))=S((h-h_T)V)^{-1}\ \ \cdot 1\cdot S((h-h_T)V+F_t)\\
&=S((h-h_T)V)^{-1}\overbrace{S((h_t-h_T)V)S((h_t-h_T)V)^{-1}}\ S((h_t-h_T)V+F_t+(h-h_t)V)\\
&=S((h-h_T)V)^{-1}S((h_t-h_T)V)S((h_t-h_T)V)^{-1}\overbrace{S((h_t-h_T)V+F_t)S((h_t-h_T)V)^{-1}S((h-h_T)V)}\\
&=\mathrm Ad(U_t)\alpha_t (S_{hV}(F))
\end{split}
\ee
with $U_t\doteq S((h-h_T)V)^{-1}S((h_t-h_T)V)$, $T>t+2\epsilon$.

$U_t$ satisfies for $s,t>0$ the cocycle relation
\be
U_{t+s}=U_t\alpha_t(U_s)\ .
\ee
The cocycle can be extended to the full real axis by defining $U_{-t}\doteq\alpha_{-t}(U_t^{-1})$.
In a representation where $\alpha_t$ is unitarily implemented by a strongly continuous 1-parameter group $W$ and where $t\mapsto U_t$ is strongly continuous, $t\mapsto U_t W(t)$ is a strongly continuous unitary 1-parameter group and its self-adjoint generator is the interacting Hamiltonian. 

Finally, we observe that for $t>0$

\be\label{eq:S(K)}
U_t=S_{(h-h_T)V}((h_t-h)V)=S_{(h-h_T)V}(\int_0^t \frac{d}{ds}h_s V ds)
=S_{(h-h_T)V}(-\int_0^t \dot{h}_s V ds)
\ee
where $\dot{h}_s$ is the time derivative of the smooth cutoff function $h$ which describes the process of switching on the interaction. It is of compact time support by construction. The generator of the cocycle $U_t$, that is, the time derivative of $U_t$ evaluated at $t=0$ and multiplied by $i$, if it exists in the representation, is denoted by $K$ and represents the interaction Hamiltonian.

We will apply this general procedure to the example of a Dirac field in an external electromagnetic field.

\subsection{Implementability of S-matrices}
The Dirac field in an external electromagnetic potential $A=(A^0,\vec A)$ can be described by the CAR algebra over the Hilbert space $\mathcal H=L^2(\mathbb{R}^3,\mathbb{C}^4)$. It is the C*-algebra generated by elements $\psi(f),f\in\mathcal H$, which depend anti-linearly on $f$, with the relation
\be
(\psi(f)+\psi(f)^*)^2=\langle f,f\rangle\ .
\ee
$\psi(f)$ is interpreted as the time-0-Dirac field smeared with a test function $f$. For the relation between time-0-Dirac fields and Dirac fields smeared with test functions supported in compact spacetime regions we refer to the paper of Dimock \cite{Dimock}.

A treatment of the CAR algebra for Fermi fields in terms of the self-dual approach due to Araki can be found in \cite{ArakiSelfDual}, see also the book \cite{ArakiBook}.

The Dirac field $\psi_{t,A}$ at time $t$ under the influence of the external potential $A$ is obtained as a solution to the equation
\be\label{eq:evolution-dirac}
i\frac{d}{dt}\psi_{t,A}(f)=\psi(i\frac{d}{dt}f_t)
\ee
where 
\be
i\frac{d}{dt}f_t=(D+\A_t)f_t
\ee
with the time dependent Dirac Hamiltonian
\be
D+\A_t=-i\vec\alpha\cdot\vec\partial+\beta m+\vec\alpha\cdot\vec A_t +A^0_t
\ee
and the electromagnetic potential $A_t$ at time $t$.
We restrict ourselves to bounded and continuous potentials $A$, then the Dirac Hamiltonian is for every $t$ described by a self-adjoint operator \cite{Kato51,Kato}. 

We want to compare the CAR algebra generated by the Dirac field which solves  \eqref{eq:evolution-dirac} with $D+A_t$, and the CAR algebra generated by the field which solves the evolution equation with the Dirac Hamiltonian $D$, namely with vanishing external potential. We call them \emph{interacting} and \emph{free} Dirac fields and denote them  by
$\psi_{t,A}(f)$ and $\psi_{t}(f)$, respectively. For a potential $A$ which vanishes at early times we require $\psi_{t,A}=\psi_t$ for $t\to-\infty$. The fields at other times are then related by
\be\label{eq:psi interacting}
\psi_{t,A}(f)=\psi_{t}(V_t(A)f)
\ee
where the unitary $V_t(A)$ is obtained by
the (operator norm converging) Dyson formula
 \be\label{eq:V_t}
V_t(A)=\sum_{n=0}^\infty i^n\int_{t_1<\dots<t_n<t}dt^n 
e^{iDt_1}\A(t_1)e^{iD(t_2-t_1)}\cdots \A(t_n)e^{-iDt_n}
\ee
with $\A(t_i)=
\vec\alpha\cdot\vec A_{t_i} +A^0_{t_i}
=\gamma_0
\gamma^\mu A_\mu(t_i,\bullet)$, considered as a matrix-valued multiplication operator. 

We now consider the Fock space over the vacuum state for the CAR Algebra generated by the fields satisfying the equation of motion with vanishing external potential.
We recall some conditions which imply that the CAR Algebra generated by fields satisfying the equation of motion with external potential can be represented by suitable operators on the Fock space of the free theory.
To this end we recall that the vacuum state $\omega$ of the free CAR algebra is given by 
\be
\omega(\psi(f_1)\cdots\psi(f_n)\psi(g_m)^*\cdots\psi(g_1)^*)=\delta_{nm}\det(\langle f_i,Pg_j\rangle)\ 
\ee
where $P$ is the projection on the spectral subspace for the positive half axis of the operator $D$.
A unitary operator $U$ on $\mathcal{H}$ induces an automorphism $\alpha_U$ of the CAR algebra by
\be
\alpha_U(\psi(f))=\psi(Uf)\ .
\ee
It can be implemented in the GNS representation $(\mathfrak{H},\pi,\Omega)$ induced by $\omega$ if and only if $[P,U]=PUQ-QUP$ is Hilbert-Schmidt (Shale-Stinespring criterion \cite{SS}), with $Q=1-P$. The unitary $\Gamma(U)$ that implements $U$ is unique up to a phase, since the representation is irreducible.
The Shale-Stinespring criterion holds for $V_\infty(A) \doteq U(A)$ if $A$ has compact support in time \cite{Ruijsenaars77,Fredenhagen,Palmer}. 

We want to obtain explicit control over the 
Hilbert-Schmidt norm of $PU(A)Q$ for a potential $\A$ that is smooth and compactly supported. We first observe that the Dyson series converges in norm, and the derivatives of $A\mapsto U(A)$ are distributions with values in $B(\mathcal{H})$, satisfying the estimate
\be
||\langle\frac{\delta^n}{\delta A^n}U(A),(A')^{\otimes n} \rangle||\le
\left(\int dt'||\A'(t')||\right)^n e^{\int dt ||\A(t)||}\ .
\ee
In order to show that $PU(A)Q$ is a Hilbert-Schmidt operator, we follow the calculation in \cite{Seipp} (cf. also \cite{Scharf}). For $QU(A)P$ an analogous argument holds.

\begin{proposition}\label{prop:U-HS}
Let $A$ be smooth and compactly supported. Then the Hilbert-Schmidt norm of $PU(A)Q$ is bounded by
\begin{equation}\label{eq:HS-estimate}
    ||PU(A)Q||_{HS}\le \frac{\pi}{\sqrt{m}}\frac{d^2}{d\lambda^2}\bigg|_{\lambda=0}e^{||A||_{I}+\lambda||\dot A||_{I}+\frac12\lambda^2||\ddot A||_{I}}
\end{equation}
with \begin{equation}
    ||A||_{I}\doteq\mathrm{max}\left(\int dt d^3k |\hat{\A}(t)|(k)\ ,\sqrt{\int d^3k(\int dt|\hat{\A}(t)|(k))^2}\right)
\end{equation}
with the spatial Fourier transform $\hat A$ of $A$ and where $|\bullet|$ denotes the norm of the matrix at each value of $k$ and $t$.
\end{proposition}
\begin{proof}
We use the Dyson series \eqref{eq:V_t} for $t=\infty$.
For the $n$-th term $U_n(A)$, we introduce new integration variables $s=(t_1+\dots +t_n)/n$ and $s_i=t_{i+1}-t_i$, $i=1,\ldots, n-1$. The Jacobian is 1, and the integration domain is $\mathbb{R}$ for $s$ and $\mathbb{R}_-$ for $s_i$, $i=1, \ldots, n-1$. 
We partially integrate twice over $s$. The boundary terms vanish due to the compactness of the support of $A$. 
We obtain for the matrix valued momentum space integral kernel of $PU_n(A)Q$ the expression
\be
(PU_n(A)Q)(p,q)=(\sqrt{p^2+m^2}+\sqrt{q^2+m^2})^{-2}\frac{d^2}{d\lambda^2}\bigg\vert_{\lambda=0}(PU_n(A+\lambda \dot{A}+\frac12 \lambda^2\ddot{A})Q)(p,q)\ .
\ee
The integral kernel of $U_n(A)$ and its derivatives
can be estimated by
\be\label{eq:25}
|\langle\frac{\delta^n}{\delta A^n}U_n(A),(A')^{\otimes n} \rangle(p,q)|\le
\left(\left(\int dt'|\hat{\A'}(t')|\right)^{*n}* \exp_*{\int dt |\hat{\A}(t)|}\right)(p-q)\ 
\ee

where $*$ denotes convolution, $\hat{\A}(t)$ is the matrix valued spatial Fourier transform of $\A(t)$.

We restrict ourselves to the case of a non-vanishing mass $m$. Then
\begin{equation}
    \int d^3p (\sqrt{p^2+m^2}+\sqrt{q^2+m^2})^{-4}\le \int d^3p(p^2+m^2)^{-2}=\frac{\pi^2}{m}
\end{equation}
and the Hilbert-Schmidt norm can be estimated in term of the $L^2$-norm of the right hand side of \eqref{eq:25} as a function of $p-q$. 
Then the Hilbert-Schmidt norm of $PU(A)Q$ is bounded by
\begin{equation}
    ||PU(A)Q||_{HS}\le \frac{\pi}{\sqrt{m}}\frac{d^2}{d\lambda^2}\bigg|_{\lambda=0}e^{||A||_I+\lambda||\dot A||_I+\frac12\lambda^2||\ddot A||_I}\ .
\end{equation}
\end{proof}

\subsection{Construction of the interacting time evolution in Fock space}

We now consider an external potential $A$ which is compactly supported in space, which is constant in time for $t\geq 0$, which vanishes at times $t<-\epsilon$ and which is smoothly switched on in the time interval $[-\epsilon,0]$. 

We construct the cocycle described in subsection 2.1. Let $A_t(t',x)\doteq A(t'-t,x)$. Then the cocycle is the product of the S-matrices $S(A-A_T)^{-1}S(A_t-A_T)$ for $T$ sufficiently large. It induces on the one-particle space the unitary 
\be\label{eq:U_t} 
U_t\doteq U(A-A_T)U(A_t-A_T)^{-1}\ ,\ t>0\ .
\ee
$U_t$ is independent of $T$ if it is chosen to be sufficiently large. $(U_t)$ is indeed a cocycle, as may be seen from the following calculation, where we exploit the independence from the choice of $T$: 
\be
\begin{split}
U_t\, e^{itD}U_s\, e^{-itD}&= U(A-A_T)U(A_t-A_T)^{-1}e^{itD}U(A-A_T)U(A_s-A_T)^{-1}e^{-itD}\\
&=U(A-A_T)U(A_t-A_T)^{-1}U(A_t-A_{T+t})U(A_{t+s}-A_{T+t})^{-1}\\&=U(A-A_T)U(A_t-A_T)^{-1}U(A_t-A_T)U(A_{t+s}-A_{T})^{-1}\\
&=U(A-A_T)U(A_{t+s}-A_{T})^{-1}=U_{t+s}\ .
\end{split}
\ee
We may split $U_t$ into factors corresponding to a splitting of the time intervals,
\be
U(A-A_T)=U(A\chi_{[-\epsilon,0]})\,U(A\chi_{[0,T-\epsilon]})U((A-A_T)\chi_{[T-\epsilon,T]})\ ,
\ee
with the characteristic function $\chi_I$ of an interval $I$, and in the same way
\be
U(A_t-A_T)=U(A_t\chi_{[t-\epsilon,t]})\,U(A\chi_{[t,T-\epsilon]})U((A-A_T)\chi_{[T-\epsilon,T]})\ .
\ee
We use the facts that
\be
U(A_t\chi_{[t-\epsilon,t]})=e^{iDt}U(A\chi_{[-\epsilon,0]})e^{-iDt}
\ee
and
\be
U(A\chi_{[0,T-\epsilon]})U(A\chi_{[t,T-\epsilon]})^{-1}=U(A\chi_{[0,t]})=e^{i(D+\A_+)t}e^{-itD}
\ee
with the time independent potential $A_+$ for positive times.
Then we get 
\be\label{eq:new-UTonH}
U_t=V_0(A)e^{i(D+\A_+)t}V_0(A)^{-1}e^{-iDt}
\ee
for the cocycle on the one-particle space with $V_0(A)=U(A\chi_{[-\epsilon,0]})$ as in \eqref{eq:V_t}.

The automorphism $\alpha_{U_t}$ is implementable in the vacuum representation of the free Dirac field, since it is a composition of two implementable automorphisms according to \eqref{eq:U_t}. Thus, also the interacting time evolution for $t>0$
can be defined in this representation. Note that this is due to the smooth switching described by the unitary $V_0$. An instantaneous switch on at time zero, corresponding to a coincidence of the free and the interacting field at time zero, would cause ultraviolet divergences that exclude the implementation of the interacting time evolution in the case of a non-vanishing spatial part of $A$. In the case of a purely electric external field, where an implementation is possible, the domains of the interacting Hamiltonian and of the free Hamiltonian have trivial intersection and therefore also there the interaction term cannot be defined as an operator.

We look at the smoothed interaction
$\mathcal{K}$ which satisfies 
\[
        V_0(A) e^{i t(D+\A_+)}
        =    e^{i t(D+\mathcal{K})}     V_0(A), \qquad t\geq0 .
\]
We want to study the operator $\mathcal{K}$ and its implementability as a selfadjoint operator on the Hilbert space $\mathfrak{H}$ of the GNS representation of the vacuum.

The operator $\mathcal{K}$ will later be used to construct the interacting equilibrium state as a vector state in the GNS Hilbert space of the KMS state of the free theory. We then can investigate the relative modular operator and Araki's relative entropy.

$\mathcal{K}$ is defined as the generator of the cocycle $U_t$
\begin{equation}
\mathcal{K}= 
-i\left.\frac{d}{dt}U_t
\right|_{t=0}=-i U(A-A_T)\left.\frac{d}{dt}U(A_t-A_T)^{-1}\right|_{t=0} 
\end{equation}
We have
\be
\left.\frac{d}{dt}U(A_t-A_T)^{-1}\right|_{t=0}=\int_{-\epsilon}^0 ds\, U((A-A_T)\chi_{[s,T]})^{-1}e^{isD}\dot{\A}(s)e^{-isD}U(A\chi_{[-\epsilon,s]})^{-1}\ ,
\ee
where we used the fact that $\dot A$ vanishes outside of the time interval $[-\epsilon,0]$, and therefore, with the factorization $U(A-A_T)=U(A\chi_{[-\epsilon,s]})U((A-A_T)\chi_{[s,T]})$ and with $U(A\chi_{[-\epsilon,s]})=V_s(A)$ according to \eqref{eq:V_t}
\be\label{eq:new-defK}
\mathcal K=\int_{-\epsilon}^0 ds\, V_s(A)e^{isD}\dot\A(s)e^{-isD}V_s(A)^{-1}\ .
\ee
We find that $\mathcal K$ is a bounded selfadjoint operator with $||\mathcal K||\le\int_{-\epsilon}^0 ds||\dot\A(s)||$.

\begin{lm}\label{lm:Kbounded}
Assume that the electromagnetic potential $A_\mu(t,\bullet)$ is smooth, of compact spatial support, constant for $t>0$ and vanishing for $t<-\epsilon$.
For the operator $\mathcal{K}$ constructed in \eqref{eq:new-defK} with $\Aslash=\gamma^0\gamma^\mu A_{\mu}$, the following estimates hold for its integral kernel in the Fourier domain,
\begin{align}\label{eq:ineq-K-H}
|\hat{\mathcal{K}}(p,q)| \leq 
H(p-q)
\end{align}
where
\begin{align}
\label{eq:defH}
H(p-q)
\doteq  
\exp_* \left(\int_{-\epsilon}^0 dt |\hat\Aslash(t)|\right) * \int dt'  |{\hat {\dot {\Aslash}}}(t')| 
* 
\exp_* \left(\int_{-\epsilon}^0 dt |\hat\Aslash(t)|\right) 
(p-q)  
\end{align}
and its $L^1$, $L^2$ norms can be estimated by 
\begin{equation}\label{eq:estimateH2}
\|H\|_1 \leq  \epsilon \lnorm{\dot\Aslash}
e^{2\epsilon \lnorm{\Aslash}}
,
\qquad
\|H\|_2 \leq  \epsilon \lnorm{\dot\Aslash}
e^{2\epsilon \lnorm{\Aslash}}
\end{equation}
where  the 
norm $\lnorm{\cdot}$ is

\begin{equation}\label{eq:lnorm}
\lnorm{A}\doteq\mathrm{max}\left(\sup_t \int d^3k |\hat{A}(t)|(k)\ ,\sup_t \sqrt{\int d^3k(|\hat{A}(t)|(k))^2}\right)\ .
\end{equation}
\end{lm}
\begin{proof}
The inequality \eqref{eq:ineq-K-H} for the integral kernel of $K$ follows from the expansion of the $V_s(A)$ given in equation \eqref{eq:V_t}, observing that $e^{i t D}$ are bounded by $1$ and $\A$ are multiplicative operators.
The estimates \eqref{eq:estimateH2} are a direct consequence of this expansion and of the observation that the support in time of $\dot{A}$ is contained in $[-\epsilon,0]$.
\end{proof}

$\mathcal{K}$ is implementable by a selfadjoint operator on the GNS representation of the vacuum  provided $P\mathcal{K}Q$ is Hilbert-Schmidt (Theorem 6.1 in \cite{Fredenhagen}). The Hilbert-Schmidt property is established in the next lemma.

\begin{lm}\label{lm:KHS}
Assume that the electromagnetic potential $A_\mu(t,\bullet)$ is smooth, of compact spatial support, constant for $t>0$ and vanishing for $t<-\epsilon$.
The operators $PU_tQ$ and $P\mathcal{K}Q$ constructed respectively out of $U_t$ in \eqref{eq:U_t} and $\mathcal{K}$ in \eqref{eq:new-defK} 
are Hilbert-Schmidt. Furthermore, 

\begin{equation}
\label{eq:HS-estimateK}
    ||P\mathcal{K}Q||_{HS}\le \frac{\pi \epsilon}{\sqrt{m}}\frac{d^2}{d\lambda^2}\bigg|_{\lambda=0}
    e^{2\epsilon \left(\lnorm{A} +\lambda \lnorm{\dot A}
    + \frac12\lambda^2 \lnorm{\ddot A}\right)}
	 \left(\lnorm{\dot A} 
    +\lambda \lnorm{\ddot A}
    + \frac12\lambda^2 \lnorm{\dddot A} \right)
\end{equation}
where the norm $\lnorm{\cdot}$ is given in \eqref{eq:lnorm}\ .
\end{lm}
\begin{proof}
We observe that $U_t=U(A-A_T)U(A_t-A_T)^{-1}$. Hence, 
$P U_t Q = P U(A-A_T) P P U(A_t-A_T)^{-1} Q + 
P U(A-A_T) Q Q U(A_t-A_T)^{-1} Q$. According to Proposition \ref{prop:U-HS} in both terms we have a product of a Hilbert-Schmidt operator and a bounded operator, hence $PU_tQ$ is Hilbert-Schmidt.

To prove that  $P\mathcal{K}Q$ is Hilbert-Schmidt, we adapt the proof of Proposition \ref{prop:U-HS} to $\mathcal{K}$ given in \eqref{eq:new-defK}.
In particular, $\mathcal{K}$ is the composition of an operator $V_s(A)$  expressed by an anti-time-ordered product, of $\dot{A}$ and of an operator expressed by a time-ordered product $V_s(A)^{-1}$.
We consider the $n$-th order contribution to the Dyson series of $V_s(A)$ and the $l$-th order contribution to the Dyson series of $V_s(A)^{-1}$ and we denote by $\mathcal{K}_{n,l}$ the corresponding contribution to $\mathcal{K}$.  
It takes the form
\begin{equation}\label{eq:Knl}
\mathcal{K}_{n,l}
= i^{n+l}\!\!\!\!\!\!\!
\int\limits_{
\substack{t_1< \dots <t_n<t' \\ 
\tau_1< \dots <\tau_{l} <t'
}}\!\!\!\!\!\!
dt'd t^{n}
d \tau^{l}  \ e^{i t_1 D}\A(t_1) e^{i (t_2 -t_1) D}    \!\!\!\dots   
e^{i (t' -t_{n}) D}  \dot\A(t') 
 e^{-i(t' - \tau_l) D}
 \!\!\!\dots
e^{-i (\tau_{2} -\tau_{1}) D}    \A(\tau_{1}) e^{-i \tau_1  D}.
\end{equation}
We observe that the domains of time integrations of the various factors composing $\mathcal{K}_{n,l}$ are contained in the intersection of the past part of the support of $\dot{A}$ and the future part of the support of $A$ which is thus contained in $[-\epsilon,0]$. 
We use now the baricenter $s$ of all time coordinates and the displacements from the baricenter as a coordinate system, as done previously in Proposition~\ref{prop:U-HS}. We observe that $s$ appears in the first operator $e^{it_1 D}$  and in the last one $e^{-i\tau_1 D}$ and in the various arguments of $\A$ and $\dot\A$. The domain of integration of $s$ is $[-\epsilon,0]$.  
We partially integrate in $s$ two times, take the sum over $n$ and $l$ and obtain 
\[
P\mathcal{K}Q(p,q)  = -\frac{1}{(\omega(p)+\omega(q))^2} \frac{d^2}{d\lambda^2}\bigg|_{\lambda=0} P\mathcal{K}(A+\lambda \dot A + \frac{\lambda^2}{2}\ddot{A})Q
\]
where we highlighted in $\mathcal{K}$ how it depends on $A$ and where $\omega(k)=\sqrt{k^2+m^2}$.
Hence, since
$\frac{1}{(\omega(p)+\omega(q))^{2}}\leq \frac{1}{\omega(p)^{2}}$, 
the Hilbert-Schmidt norm of  $P\mathcal{K}Q$ can be estimated by the product of the $L^{2}$ norm of $1/\omega(p)^2$ which is $\pi/\sqrt{m}$ and the $L^2$ estimates of $\mathcal{K}$ given in \eqref{eq:estimateH2} of Lemma \ref{lm:Kbounded}. 
In particular, we get the non-optimal estimates given in the thesis.
\end{proof}

The estimate obtained in Lemma \ref{lm:KHS} grows  exponentially in $\epsilon$ where $\epsilon$ measures the time extension of the support of $\dot\A$.
To study the case of $A$'s that are switched on adiabatically, and thus very slow, we need to improve that estimate. 
This task is accomplished in the following at the price of requiring $A$ to be sufficiently small.
\begin{lm}\label{lm:KHSunifrom}
Consider $\A_{\mathcal{T}}(t)\doteq h_\mathcal{T}(t)\A$, $\mathcal{T}>0$ where $\A$ is smooth, time independent with compact spatial support, and $h_\mathcal{T}(t)=h(t/\mathcal{T})$ is a time cutoff function constructed from a smooth function $h$ which is equal to $0$ for $t<-1$ and equal to $1$ for $t>0$, and with a nonnegative time derivative $\dot{h}\geq 0$.
If $\A$ is sufficiently small, 
the operator $\mathcal{K}$ corresponding  to $\A_{\mathcal{T}}$ satisfies the estimate
\[
\|P\mathcal{K}Q\|_{HS}
\leq  \frac{C}{\mathcal{T}^2} 
\]
for some positive constant $C$ independent of $\mathcal{T}$,
and hence the Hilbert-Schmidt norm of $P\mathcal{K}Q$ vanishes in the limit $\mathcal{T}\to \infty$.
\end{lm}
\begin{proof}
With the proof of Proposition \ref{prop:U-HS} and of Lemma \ref{lm:KHS} in mind 
we look for an efficient bound
for the time integrations in $\mathcal{K}_{n,l}$. 

We start considering the case $\mathcal{T}=1$. We observe that
in the hypotheses of the present Lemma contributions involving $h$ and $\A$ factorize in the integrand of the momentum space integral kernel of $P\mathcal{K}_{n,l}Q$
\[
P\mathcal{K}_{n,l}Q(p,q)=-\frac{i^{n+l}}{(\omega(p)+\omega(q))^2}
\int_{\mathcal{D}}
dx^{n+l+1} K_{n,l}(x) \mathsf{H}_{n,l}(x)
\]
with $x=(t,t',\tau)\in\mathbb{R}^{n+1+l}$ and the domain of integration $\mathcal{D} = \{x |x_1< \dots <x_{n+1},
x_{n+1}>\dots > x_{n+1+l} \}$. 
$\mathsf{H}_{n,l}$ is the $L_1$ matrix norm (the sum of the absolute values of the entries) of the Hessian of the function $ h(x_1) \dots \dot{h}(x_{n+1})\dots h(x_l)$,
and $K_{n,l}(x;p,q)$
is the integral kernel of the operator
$
K_{n,l}(x)
=
P
 e^{ix_1 D}\A e^{i (x_2 -x_1) D}   \dots 
 \!\!\!    \A 
e^{-i x_{n+l+1} D} Q.
$
With respect to the proofs of Proposition \ref{prop:U-HS} and of Lemma \ref{lm:KHS}, we look for an extra time decay in $K_{n,l}$.
The momentum space integral kernel of $K_{n,l}(x)$ is the matrix valued function 
\[
K_{n,l}(x)(p,q) = \! \int \!\! dp^{n+l}
e^{i\omega(p)x_1}
\prod_{j=1}^{n+l}\left(
\sum_{\sigma_j \in \{\pm\}}
\hat{\A}(p_j-p_{j-1})
P_{\sigma_j}(p_j)e^{\sigma_j i(x_{j+1}-x_j)\omega(p_j)}
\right)
\hat{\A}(q-p_{n+l})
e^{i\omega(q)x_{n+l+1}}
\]
where $p_0=p$ and $p_{n+l+2}=q$, furthermore,  $P_+ = P$ and $P_-=Q=1-P$ are matrix valued functions of spatial momentum of the form
$ 
P_\pm(p)= (\mp \omega(p)-\alpha^ik_i -\beta m)(2 \omega(p))^{-1}
$, and their matrix norms $|P_\pm(p)|$ are bounded by $1$.

Let us analyze the integral over $p_j$ in 
\[
I(p,q;t)=\int dp_j\hat{\A}(p-p_j)e^{\pm i t \omega(p_j)}\hat{\A}(p_j-q)
\]
where $\omega(p_j)=\sqrt{p_j^{2}+m^2}$.
For small times $I$ is bounded by the convolution of $|\hat \A|*|\hat{\A}|(p-q)$.
Furthermore, according to Lemma \ref{lm-decay} since $\A$ is of compact spatial support, we find that for large times $|I(p,q,t)|$ decays as $c(|\hat{A}(p)||\hat{A}(q)|+|\hat{A}|*|\hat{A}|(p-q))|t|^{-3/2}$ where the constant $c$ depends on the extent of the domain of $A$.

We can combine these observations in the following  inequality  valid for every $t$ and for some positive constant $c$
\[
|I(p,q;t)| \leq 
\frac{c}{(1+|t|)^{3/2}} 
\left(\frac{|\hat{\A}|*|\hat{\A}|(p-q) + |\hat{\A}(p)||\hat{\A}(q)|
}{2}\right).
\]
Using it for all the $p_i$ integrations we get
\begin{align*}
|P\mathcal{K}_{n,l}Q(p,q)|
\leq
&  
\int_\mathcal{D}dx^{n+l+1}
\prod_{j=1}^{n+l}\frac{c}{(1+|x_i-x_{i+1}|)^{\frac{3}{2}}}
\mathsf{H}_{n,l}(x)
\left|\frac{\mathbf{A}_{n,l}(p,q)}{(\omega(p)+\omega(q))^2}\right|\ .
\end{align*}

Furthermore,
$\mathbf{A}_{n,l}(p,q)$
contains $n+l+1$ $\A$ factors, sometimes convoluted with themselves, sometimes evaluated at $0$ or $p$ or $q$.
We have that 
\[
\left\|\frac{1}{(\omega(p)+\omega(q))^2}\mathbf{A}_{l,n}(p,q)\right\|_{2}
\leq \frac{\pi}{\sqrt{m}} \lnorm{\A}^{n+l+1}
\]
where the norm $\|\cdot\|_2$ is for a function on $L^2(\mathbb{R}^3\times\mathbb{R}^3)$ and
where the norm $\lnorm{\cdot}$ is given in \eqref{eq:lnorm} and here it is applied to $\A$ which is constant in time, furthermore, since $A$ is of compact support we used the fact that $\| \hat{A} \|_\infty \leq \|A\|_1  \leq c \|\hat{A}\|_2$ for some constant $c$ depending on the support of $A$ and where $||\cdot||_\infty$ is the $L^\infty$-norm.

Using the following non optimal bounds  
\[
\int_0^t ds \frac{1}{(1+t-s)^{\frac{3}{2}}} 
\frac{1}{(1+s)^{\frac{3}{2}}} 
\leq \frac{4}{(1+t)^{\frac{3}{2}}}, 
\qquad
\int_{0}^{\infty} ds \frac{1}{(1+s)^{\frac{3}{2}}}=2 
\]
and bounding the contributions of $h$ or of its derivative in $\mathsf{H}_{n,l}$ by the sup norms, the integrals over $x_i$ with $i\neq n+1$ can be taken. 
The integral over $x_n$ which remains is estimated by $\|\dot{h}\|_1$, by $\|\ddot{h}\|_1$ or by $\|\dddot{h}\|_1$. 
Observing that $\|h\|_\infty = 1$, so that it can be ignored, we have that the sums of $n$ and $l$ in $P\mathcal{K}Q$ are geometric sums which can be taken if $\lnorm{A}$ is sufficiently small. 
We obtain for some constant $c'$ and $c''$  
\begin{equation*}
\|P\mathcal{K}Q\|_{HS}
\leq \frac{c'}{4}\left.\left(
\|\dot{h}\|_1
\|\dot{h}\|_\infty^2
\frac{d^2}{d\lambda^2}
+ 
\|\dot{h}\|_1
\|\ddot{h}\|_\infty
\frac{d}{d\lambda}
+
2\|\ddot{h}\|_1
\|\dot{h}\|_\infty 
\frac{d}{d\lambda}
+
\|\dddot{h}\|_1
\right)
 \frac{\lnorm{\A}}{\left( 1- c'' \lambda  \lnorm{\A}\right)^{2}}\right|_{\lambda=1}
\end{equation*}
where the derivatives in $\lambda$ are present in this formula in order to take care of the $n,l$ dependence in Hessian in $\mathsf{H}_{n,l}$. Furthermore, we have used the fact that the sums over $n$ and $l$ are geometric series or theirs derivatives.
If we now substitute $h_\mathcal{T}$ in place of $h$ we have for some positive constant $C$
\[
\|P\mathcal{K}Q\|_{HS}
\leq  \frac{C}{\mathcal{T}^2} 
\]
and this vanishes in the limit $\mathcal{T} \to \infty$.
\end{proof}

\section{Equilibrium states,  the Cook criterion and stability properties} \label{se:equilibrium-states}

States at finite temperature are characterized by the celebrated Kubo-Martin-Schwinger condition (KMS) \cite{HHW}. 
We recall in particular that in the case of a $C^{*}$-dynamical systems $(\mathfrak A,\tau)$,
a state $\omega$ is KMS with respect to $\tau$ at inverse temperature $\beta$ if $\omega$ is invariant under $\tau$ and if for every $B,C\in\mathfrak{A}$ there exists a continuous function $F$ on the strip $\Im{t}\in[0,\beta]$ which is analytically in the interior and satisfies $F(t)=\omega(B\tau_t(C)))$ and $F(t+i\beta)=\omega(\tau_t(C)B)$, see e.g. \cite{BratteliRobinson}.  
This condition can be generalized also for the case of $*$-algebras where conditions on higher correlation functions have to be added which can be derived for ${C^*}$-dynamical systems.  

For the Dirac field theory with vanishing external potential, the equilibrium state $\omega^{\beta,\mu}$ at the inverse temperature $\beta$ and the chemical potential $\mu$, with respect to the time evolution $\tau^\mu_t$, 
\[
\tau^\mu_t(\psi(f))= \psi(e^{it(D-\mu)}f)\ ,
\]
is a quasi-free state whose two-point function can be obtained by means of the KMS condition.
The KMS condition together with the anticommutation relations  implies that 

\begin{align*}
\omega^{\beta,\mu}(\psi(f)\psi(g)^*)+\omega^{\beta,\mu}(\psi(f)\psi(e^{-\beta (D-\mu)}g)^*)
& =
\langle f,g\rangle.
\end{align*}
An analogous relation holds for $\omega^{\beta,\mu}(\psi(g)^*\psi(f))$. Hence,

\begin{equation}\label{eq:two-point-kms}
\omega^{\beta,\mu}(\psi(f)\psi(g)^*)
=
\langle 
f, \frac{1}{1+e^{-\beta (D-\mu)}}g\rangle, 
\qquad 
\omega^{\beta,\mu}(\psi(g)^*\psi(f)) 
=
\langle 
f, \frac{1}{1+e^{\beta (D-\mu)}}g\rangle.
\qquad 
\end{equation}

The corresponding $n$-point functions are
\be
\omega^{\beta,\mu}(\psi(f_1)\cdots\psi(f_n)\psi(g_m)^*\cdots\psi(g_1)^*)=\delta_{nm}\det(\langle f_i,\frac{1}{1+e^{-\beta (D-\mu)}}g_j\rangle)\ .
\ee
We observe that the two-point function of the KMS state with non-vanishing chemical potential can be obtained as the KMS state with vanishing chemical potential by replacing the operator $D$ with $D_\mu = D-\mu$.
For this reason, we shall restrict our analysis to the case of vanishing chemical potential. For sufficiently small $\mu$ it is possible to recover the case with a non-vanishing chemical potential substituting $D$ with $D_\mu$ in the results we present.

\subsection{Analytic continuation of the cocycle and the state of Araki}

In the context of $C^{*}$-dynamical systems $(\mathfrak A,\tau)$,  with a C*-algebra $\mathfrak{A}$ and a strongly continuous 1-parameter group $(\tau_t)$ of automorphisms describing the time translations, perturbations of the time evolution and of KMS-states can be treated in the following way. Let $K\in\mathfrak{A}$. The Dyson series for $K(t)=\tau_t(K)$ yields a continuous unitary cocycle $(U_t)$ and a perturbed time evolution $\tau^K_t\doteq \mathrm{Ad}(U_t)\tau_t$. 

Let $\omega_\beta$ be a $\beta$-KMS state with respect to $\tau$. An equilibrium state for the interacting dynamics $\tau^K$ can be constructed in the following way \cite{Araki-KMS, BKR, BratteliRobinson}.

We consider the expansion
\[
\omega^{\beta}(BU(t)) =  \omega^{\beta}(B)  +  \sum_{n} i^{n}\int_{0<s_1<\dots <s_n<t}\!\!\!\! dS\; \omega^{\beta}(B \tau_{s_1}(K) \dots \tau_{s_n}(K)), \qquad B \in \mathfrak{A}.
\] 
where $S=(s_1,\dots ,s_n) \in \mathbb{R}^{n}$ and $dS=ds_1\dots ds_n$.
 Due to the KMS condition, the integrand on the right-hand side can be analytically extended to the region $\{0<\Im (s_1)<\dots<\Im (s_n)<\Im  (t)<\beta\}\subset \mathbb{C}^n$, with bounded and continuous boundary values.
Hence, the function
\[
F(t)\doteq \omega^{\beta}(BU(t)), \qquad B \in \mathfrak{A},
\]
 is analytic for $\Im (t) \in (0,\beta)$ and continuous and bounded at the boundary (see \cite{Araki-KMS} for the case of von Neumann algebras or \cite{BKR, BratteliRobinson}). 

This implies that an equilibrium state (KMS) for the perturbed dynamics 
\[
\tau_t^{K}(B) = U(t) \tau_t(B) U(t)^*
\]
at inverse temperature $\beta$ is obtained as \cite{Araki, FL}
\begin{equation}\label{eq:araki-state}
\omega^{\beta K}(B) = \frac{\omega^\beta(B U_{i\beta})}{\omega^\beta(U_{i\beta})}.
\end{equation}
It is furthermore proved in \cite{BKR} that
\begin{equation}\label{eq:omegauib}
\omega^{\beta K}(B) =  \omega^{\beta}(B)  +   \sum_{n} (-1)^n\int_{0<u_1\dots <u_n<\beta} dU \omega_c^{\beta}(B, \tau_{i u_1}(K), \dots ,\tau_{i u_n}(K))
\end{equation}
where $dU=du_1\dots du_n$ and where now $\omega^\beta_{c}$ are the connected correlation functions obtained from $\omega^\beta$.

An analogous construction of equilibrium states was performed in the framework of formal perturbation theory \cite{FL}, where $C^*$-algebras have to be replaced by $*$-algebras of formal power series. 
Also there \eqref{eq:omegauib} holds.

In the case analyzed in the present paper we are in an intermediate situation. The perturbation $K$ is an unbounded operator in the GNS-representation of the KMS-state, and the associated cocycle is not contained in the CAR-algebra. We can, however, exploit the fact that both arise from bounded operators $\mathcal K$ and $U_t$, respectively, in the one-particle Hilbert space.

We showed that both $U_t$ and $\mathcal K$ can be implemented in the vacuum representation,
$U_t$ by a unitary operator and $\mathcal K$ as an unbounded  selfadjoint operator $K$. Later we will show that analogous properties hold also in the GNS-representation of KMS states (see Appendix). 

We use formula \eqref{eq:omegauib} also in the case under investigation in this paper and show that it results in an operator on the one-particle Hilbert space that intertwines the free and interacting KMS state.

We recall that the interacting Hamiltonian $K$ is quadratic in the fields, and that it is the implementation of the operator $\mathcal{K}$ given in equation \eqref{eq:new-defK} above. 

The operator $K$  takes the form \cite{Fredenhagen}
\begin{equation}\label{eq:int-ham}
K = -\sum_i (\psi( f_i) \psi(\mathcal{K}f_i)^* +c_i)
\end{equation}
where $f_i$ are the elements of an orthonormal basis of the one-particle Hilbert space $\mathcal{H}$, and where the constants $c_i$ correspond to normal ordering. 
Since multiples of 1 vanish in truncated functionals of order $>1$, these constants do not contribute to \eqref{eq:omegauib}.

For the case of a quadratic functional $B = \psi(f)\psi(g)^*$, a particularly simple representation for $\omega^{\beta K}(B)$ can be obtained,
\begin{gather}\label{eq:expomegabetav1}
 \omega^{\beta K}(\psi(f)\psi(g)^*) = 
 \langle f,  F(0) g \rangle
 -
 \!\!
 \sum_{n>0}\sum_{\pi\in \mathcal{P}_{n}}
   \!(-1)^n\!\!\!\!\!\!\!\!\!\!\int\limits_{0<u_1\dots <u_n<\beta}\!\!\!\!\!\!\!\!\!\! d^n u 
    \;
   \langle f, F(u_{\pi_1}) 
  \left(
  \prod_{j=1}^{n-1}
  \mathcal{K}
  F(u_{\pi_{j+1}}-u_{\pi{j}}) 
  \right)
   \mathcal{K}
  F(\beta - u_{\pi_n}) 
   g\rangle 
 \end{gather}
 where the contribution of the product $\prod_{j=1}^{n-1}$ for $n=1$ is assumed to be $1$. Furthermore,
$F$ is an operator valued function on $\mathbb{R}$ which is antiperiodic of period $[0,\beta)$ and which is defined as  
\begin{equation}\label{eq:fermi-factors}
F(u) \doteq \frac{e^{ - u D}}{1+e^{-\beta D}}, 
\end{equation}
for $u\in[0,\beta)$.
Notice that for every $u$, $F(u)$ is a bounded operator.
  We have also that $F_- \doteq (1+e^{-\beta D})= F(0)$ and $F_{+} \doteq -(1+e^{\beta D})=F(0)-1$ are called {\bf Fermi factors} of the free theory.
Finally, $\mathcal{P}_{n}$ contains all the permutations of the elements of the set $\{1,\dots, n\}$.
To obtain the formula, we used the two-point functions given in eq. \eqref{eq:two-point-kms} and the fact that, for factors which are quadratic in the fields, the truncated correlation functions with $n$ factors can be given as a sum over the connected graphs joining $n$ vertices with $n$ edges ({\it loops}).
One of the vertices corresponds to the rank 1 operator $|g\rangle\langle f|$, the others to $\mathcal{K}$ possibly translated in imaginary time, and the edges to $F_{\pm}$, and the graph contributes by the trace of the product of operators along the loop.
The sum over $\mathcal{P}_{n}$ in \eqref{eq:expomegabetav1} can be resolved by enlarging the domain of integration from  $u\in \beta{\mathcal{S}_n}=\{(u_1,\dots ,u_n) \in \mathbb{R}^{n} | 0<u_1<\dots <u_n<\beta\}$ to $u\in (0,\beta)^n$.  We obtain 
\begin{gather}\label{eq:expomegabetav}
\omega^{\beta K}(\psi(f)\psi(g)^*) = 
\langle f,  \frac{1}{1+e^{-\beta D}} g \rangle
-
 \sum_{n>0}
   (-1)^n\!\!\!\!\!\int\limits_{u\in (0,\beta)^n}\!\!\!\!\! d^n u 
    \;
   \langle f, F(u_1) 
  \left(
  \prod_{j=1}^{n-1}
  \mathcal{K}
  F{(u_{j+1}-u_j)} 
  \right)
   \mathcal{K}
  F(\beta-u_n) 
   g\rangle \ .
\end{gather}
A similar expression for the case of Klein Gordon fields with perturbation of the mass has been found in \cite{Drago}.
We discuss in the following the well-posedness of the previous expression and that it satisfies the KMS condition with respect to the interacting time evolution. 

\begin{thm}\label{lm:T}
Consider
\begin{equation}\label{eq:def-T}
T(\beta) \doteq   \frac{1}{1+e^{-\beta D}}  
-
 \sum_{n>0}
   (-1)^n\!\!\!\!\!\int\limits_{u\in (0,\beta)^n}\!\!\!\!\! d^n u 
    \;
   F(u_1)  
  \left(
  \prod_{j=1}^{n-1}
  \mathcal{K}
  F(u_{j+1}-u_j)
  \right)
   \mathcal{K}
  F(\beta - u_n)\ 
\end{equation}

If $\|\mathcal{K}\| <\beta^{-1}$, 
the series which defines $T$ is absolutely convergent and converges to a bounded operator in the uniform operator topology.  

Furthermore, it holds that 
\begin{equation}\label{eq:Tbeta}
T(\beta) =  \frac{1}{1+e^{-\beta (D+\mathcal{K})}}
\end{equation}
and thus, recalling \eqref{eq:expomegabetav}
\[
\omega^{\beta K}(\psi(f)\psi(g)^*) = 
\langle f,  T(\beta) g \rangle
\]
is the two-point function of the unique KMS state with respect to the time evolution $\tau_t^K$ at inverse temperature $\beta$.
\end{thm}
\begin{proof}
We have that for every $u\in (0,\beta)$ the operators $e^{u D} (1+e^{\beta D})^{-1}$ and $e^{-u D} (1+e^{-\beta D})^{-1}$ have norm smaller than $1$. Hence, 
recalling the form of $F$ and considering the domain of integration of $u$, we have
\[
\|T(\beta)\| \leq \sum_{n\geq 0}  \beta^{n} \|\mathcal{K}\|^n
\leq
\frac{1}{1-\beta \|\mathcal{K}\|}
\]
and thus the series is absolutely convergent if $\beta\|\mathcal{K}\|<1$.

We now check that $T(\beta)$ has the form in \eqref{eq:Tbeta} and
 to this avail, we use the
recursive relation for $T(\beta)$ given in \eqref{eq:def-T} and prove that the recursion has a vanishing reminder in the limit for a large number of steps. 
We start observing that $-1$ is in the resolvent set of both the unbounded operators $e^{-\beta (D+\mathcal{K})}$ and $e^{-\beta D}$ because both operators are positive. 
Hence, for $u\in (0,\beta)$ we use a modified  second resolvent equation to analyze
\begin{equation}
\begin{aligned}
\frac{e^{-u (D+\mathcal{K})}}{1+e^{-\beta (D+\mathcal{K})}} &-
\frac{e^{-u D}}{1+e^{-\beta D}}
\\
&=
\frac{1}{1+e^{-\beta (D+\mathcal{K})}}
\left(e^{-u (D+\mathcal{K})}-e^{-u D} +
e^{-u (D+\mathcal{K})}e^{-\beta D}
-e^{-\beta (D+\mathcal{K})}e^{-u D} \right)
\frac{1}{1+e^{-\beta D}} 
\\
&= 
-\frac{1}{1+e^{-\beta (D+\mathcal{K})}}
\int_0^\beta 
e^{-\overline{u} (D+\mathcal{K})}
\mathcal{K}e^{\overline{u} D}
\left(
\theta(u-\overline{u}) - e^{-\beta D} \theta(\overline{u}-u)
\right)
d\overline{u} \; 
\frac{e^{-u D}}{1+e^{-\beta D}} 
\end{aligned}
\label{eq:recursive-relations}
\end{equation}
where $\theta$ is the Heaviside step function. To obtain this equation, we used the relation
\[
e^{- u (D+\mathcal{K})}e^{ uD} -1 = -\int_0^{u} e^{-\overline{u}(D+\mathcal{K})}\mathcal{K} e^{\overline{u}D}
d\overline{u}.
\]
Notice also that, after factoring out, all operators present on the right-hand side of \eqref{eq:recursive-relations} are bounded. 
Actually, the integral is taken over a finite interval, and for every value of the integration variable the 
integrand has the form of 
\[
\frac{e^{-a (D+\mathcal{K})}}{1+e^{-\beta (D+\mathcal{K})}} \mathcal{K}
\frac{e^{-b D}}{1+e^{-\beta D}} 
\]
for some $a,b\in(0,\beta)$.
From \eqref{eq:recursive-relations} we get a relation that can be used recursively,
\begin{equation}\label{eq:recursive-relation-simplified}
\frac{e^{-u (D+\mathcal{K})}}{1+e^{-\beta (D+\mathcal{K})}}  
= 
\frac{e^{-u D}}{1+e^{-\beta D}}
-
\int_0^\beta 
\frac{e^{-\overline{u} (D+\mathcal{K})}}{1+e^{-\beta (D+\mathcal{K})}}
\mathcal{K}e^{(\overline{u}-u) D}
F_{\sigma(\overline{u}-u)}
d\overline{u} \; 
\end{equation}
where $\sigma$ is the sign function and
$F_\pm$ are the Fermi factors given in \eqref{eq:fermi-factors}.
Using recursively this relation, we obtain exactly \eqref{eq:def-T} for $T(\beta)$ \footnote{The expansion in powers of $\mathcal K$ can also be obtained directly (Proposition \ref{prop:derivatives} in Appendix \ref{se:derivatives}).}
. 
Notice that if we truncate the series after $n$ steps, the error we commit is estimated in norm by $\beta^n\|\mathcal{K}\|^{n}$ which tends to $0$ in the limit of large $n$.

We finally observe that 
\[
{\omega}^{\beta K}(\psi(f)\psi(g)^*) = \langle f, \frac{1}{1+e^{-\beta (D + \mathcal{K})}} g \rangle \ .
\]
It provides the two-point function of a KMS with respect to the time evolution $\tau_t^{K}$ at the inverse temperature $\beta$. 
\end{proof}

We conclude this section recalling that there exists a unique KMS state for the evolution generated by $D+\mathcal{K}$. The state we have obtained is thus the unique quasifree state which is KMS with respect to the time evolution $\tau^K_t$.

\section{Approach to equilibrium and non-equilibrium steady states} \label{se:NESS}

Stability properties of KMS states have been considered by many authors. Remarkable results 
are usually obtained when the interacting Hamiltonian $K$ is an element of the $C^*$-algebra and if some clustering conditions hold \cite{BKR, HaagTrych, HKT}, see also \cite{BratteliRobinson}.

Recently some of these results have been extended to the case of perturbation theory where 
stability holds in a weaker sense  \cite{DFP}. 
In the case under investigation in this paper we have the one-particle structure well under control. Since we  consider perturbations which are quadratic in the field, the results obtained in the one-particle Hilbert space can be used to analyze stability properties of the full system.

To be more specific we would like to clarify under which conditions
\[
\lim_{t\to\infty} \omega^\beta(\tau_t^{K}(B)) =  \omega^{\beta K}(B)\ .
\]
If this holds we say that the KMS state $\omega^\beta$ is stable under perturbations. Furthermore we are also interested in return to equilibrium properties, namely, limits of the following way
\[
\lim_{t\to\infty} \omega^{\beta K}(\tau_t(B)) =  \omega^{\beta }(B)\ .
\]
If the latter holds we say that there is return to equilibrium.
The results presented in this section make use of the Cook criterion (see e.g. Section X.3.2 in \cite{Kato}) to estimate the stability and return to equilibrium properties of the free and interacting KMS state. 

Since the interaction Lagrangian we are considering is quadratic in the field, it is sufficient to study these problems at the level of the two-point function (the states we obtain remain quasifree). For the stability problem, we want to obtain conditions under which 
\[
\lim_{t\to\infty} \omega^\beta(\tau_t^{K}(\psi(f)\psi(g)^*)) =  \lim_{t\to\infty} \omega^\beta(\tau_{-t}\tau_t^{K}(\psi(f)\psi(g)^*)) = 
\omega^{\beta K}(\psi(f)\psi(g)^*)\ .
\]
We start observing that 
\[
\lim_{t\to\infty} \omega^\beta(\tau_t^{K}(\psi(f)\psi(g)^*)) =  \lim_{t\to\infty} \omega^\beta(\tau_{-t}\tau_t^{K}(\psi(f)\psi(g)^*)) =  \lim_{t\to\infty} \langle {U}_t^* f, \frac{1}{1+e^{-\beta D}} {U}_t^* g \rangle
\]
where we used ${U}_t = e^{i t (D+\mathcal{K})}e^{-itD}$ given in eq. \eqref{eq:new-UTonH} above. 

Hence the analysis can be done at the level of the one-particle Hilbert space and we are looking for some conditions under which  
\[
\lim_{t\to\pm\infty} {U}_t \frac{1}{1+e^{-\beta D}} {U}_t^{*}  = \frac{1}{1+e^{-\beta (D+\mathcal{K})}} \ . 
\]
We shall analyze this limit by using the Cook-criterion.

We consider the quantum mechanical Möller operators for $D$ and $D+\mathcal{K}$ on the Hilbert space $\mathcal{H}$. We assume that $D+\mathcal{K}$ has absolutely continuous spectrum and that the Cook-criterion holds,
\be
\int dt ||\mathcal{K}e^{it D}f||<\infty\ ,
\ee
for $f$ in a dense subspace. Then the Möller operators
\be
\Omega_\pm=\mathrm{s}-\lim_{t\to \pm\infty}e^{it(D+\mathcal{K})}e^{-it D}
\ee
exist and are isometries (see e.g. Theorem 3.7 in Chapter X of \cite{Kato}). But then
\be
e^{itD}(1+e^{\beta(D+\mathcal{K})})^{-1}e^{-itD}\to\Omega_+^*(1+e^{\beta(D+\mathcal{K})})^{-1}\Omega_+
\ee
in the weak operator topology. Since $\Omega_+f(D)=f(D+\mathcal{K})\Omega_+$ for continuous bounded functions $f$ we see that the KMS state with external potential converges to the free one, \ie return to equilibrium holds for $\omega^\beta$.

If the Cook-criterion also holds for $D+\mathcal{K}$ instead of $D$ we get also the convergence from the interacting to the free KMS state, \ie $\omega^\beta$ is stable under these perturbations.

\begin{lm}\label{lm:Cook-criterion-D}
If ${A}$ is of compact spatial support, the Cook-criterion holds for $D$.
\end{lm}
\begin{proof}
If the vector potential ${A}$ is of spatially compact support, the interacting Hamiltonian is also of compact spatial support. The same holds for $\mathcal{K}$ given in \eqref{eq:new-defK}. We can thus find a smooth compactly supported function $\chi$ defined on $\Sigma_0$ which is equal to the identity  on the support of $\mathcal{K}$. 
Furthermore, the compactly supported smooth functions $C^{\infty}_c(\Sigma_0;\mathbb{R}^{4})$ are dense in $\mathcal{H}$, hence for ever $f\in C^{\infty}_c(\Sigma_0;\mathbb{R}^{4})$, it holds that 
\[
\| \mathcal{K}e^{itD} f \| =
\| \mathcal{K} \chi e^{itD} f \| \leq
\| \mathcal{K}\|  \|\chi e^{itD} f \|
\leq C  \|\chi e^{itD} f \|
\]
where in the last step we used that $\mathcal{K}$ is bounded. 
By stationary phase methods, we have that $\|\chi e^{itD} f\|$ decays for large $t$ at least as $1/t^{3/2}$ hence the thesis holds.
To prove this decay we observe that 
\[
\reallywidehat{(\chi e^{itD} f)}(p) 
= \int d^3{k} 
\overline{\chi}(p-k) 
\left((\omega(k)-\alpha^ik_i -\beta m)\frac{e^{-i t \omega(k)}}{2\omega(k)} 
+
(-\omega(k)-\alpha^ik_i -\beta m)
\frac{e^{i t \omega(k)}}{2\omega(k)}
\right) 
\hat{f}(k)
\]
where $\omega(k)=\sqrt{k^2+m^2}$.
To prove that the $L^2$ norm of the previous function decays we use the results of Lemma \ref{lm-decay}.
\end{proof}

\begin{lm}\label{lm:Cook-criterion-D+K}
If the vector potential ${A}=\lambda \mathcal{A}$ is of compact spatial support and if $\lambda$ is sufficiently small so that $\|\mathcal{K}\|$ is also sufficiently small, the Cook-criterion holds also for $D+\mathcal{K}$.
\end{lm}
\begin{proof}
Writing $e^{i t (D+\mathcal{K})}$ as a Dyson series we have
\[
\mathcal{K}e^{i t(D+\mathcal{K})}
f
= 
\sum_{n} i^n
\int_{t \mathcal{S}_n} 
dt^{n}
\mathcal{K}
e^{i (t-t_n) D} \mathcal{K}
e^{i (t_n-t_{n-1}) D}
\dots e^{i (t_2-t_1) D} \mathcal{K}
e^{i t_1 D} f
\]We operate as in the previous Lemma.
Using an $\chi$ which is smooth, compactly supported and equal to the identity on support of $\mathcal{K}$, we have that 
\[
\|\mathcal{K}e^{i t(D+\mathcal{K})}
f\|
\leq 
\sum_{n} 
\|\mathcal{K}\|^{n+1}
\int_{t \mathcal{S}_n} 
dt^{n}
\| \chi e^{i (t-t_n) D} \chi\| 
\dots \|\chi e^{i (t_2-t_1) D}\chi\|
\dots \| \chi e^{i (t_1) D} f\|
\]
with stationary phase methods similar to those discussed in the proof of Lemma \ref{lm:Cook-criterion-D} we get the $1/t^{3/2}$ decay for every norm in the integrand, more precisely, we have that it holds
\[
\|\chi e^{itD} \chi \| \leq \frac{C}{(1+t)^{3/2}}, \qquad
\|\chi e^{itD} f \| \leq \frac{C_f}{(1+t)^{3/2}}
\]
for $t\geq 0$. 
Hence 
\[
\|\mathcal{K}e^{i t(D+\mathcal{K})}f\|
\leq 
\sum_{n} 
C \|f\| C^{n}\|\mathcal{K}\|^{n+1}
\int_{t \mathcal{S}_n} 
dt^{n}
\frac{1}{(1+t_1)^{\frac{3}{2}}}
\frac{1}{(1+t_2-t_1)^{\frac{3}{2}}}
\dots 
\frac{1}{(1+t-t_n)^{\frac{3}{2}}}
\]
Using recursively the bound  
\[
\int_0^t ds \frac{1}{(1+t-s)^{\frac{3}{2}}} 
\frac{1}{(1+s)^{\frac{3}{2}}} 
\leq \frac{4}{(1+t)^{\frac{3}{2}}} 
\]
we get a geometric series which can be summed if $4C\|\mathcal{K}\|<1$. In that case we obtain 
\[
\|\mathcal{K}e^{i t(D+\mathcal{K})}f\|
\leq 
 \frac{4 C_f C\|\mathcal{K}\|}{1- 4C\|\mathcal{K}\|} \frac{1}{(1+t)^{\frac{3}{2}}}
\]
and the Cook-criterion holds.
\end{proof}

\bigskip

We comment now on the fact that the absolute continuity of the spectrum of $D+\mathcal{K}$ is an essential request to have return to equilibrium. Actually, if there are bound states for $D+\mathcal{K}$,  return to equilibrium and or stability of KMS states are not expected to hold.
To make this manifest, we consider an external potential $\A$ which is strong enough to make the point spectrum of the operator $D+\mathcal{K}$ not empty and such that it contains at least an isolated eigenvalue $s$ (disjoint from the continuous spectrum of $D+\mathcal{K}$). 
Consider $\varphi_s$ one of the normalized eigenstates of $D+\mathcal{K}$ with eigenvalue $s$ (which is thus an element of the the point spectrum, we assume that there is no degeneracy, namely that the eigenspace of $s$ has dimension $1$).

Consider $P_s$ the orthogonal projection to the eigenspace corresponding to $s$.
The interacting time evolution applied to the $(1+e^{-\beta D})$ and projected to the eigenspace corresponding to $s$ gives
\[
O_s(t)=P_s e^{i t (D+\mathcal{K})}\frac{1}{1+e^{-\beta D}}e^{-i t (D+\mathcal{K})}P_s =
P_s\frac{1}{1+e^{-\beta D}}P_s.
\]

We have that, if at time $0$, $O_s$ is not of the form
\[
O_s \neq \frac{1}{1+e^{-\beta s}}
\]
the state we obtain at large $t$ cannot be an equilibrium state because $O_s$ cannot change in time.

In the generic case, we do not expect that the occupation number in the eigenspace of the eigenvalue $s$ of $D+\mathcal{K}$ of the equilibrium state wrt $D$ has the Fermi form, hence return to equilibrium cannot hold when $D+\mathcal{K}$ has bound states corresponding to isolated eigenvalues.

We remind the reader some   conditions on the potential which guaranties the appearance of bound states in the spectral gap region $(-m,m)$
for $D+\A$.
We refer in particular to the work of Klaus \cite{Klaus}, see also the book of
Thaller \cite{Thaller}. If $\vec{A}=0$ and $A^0 = \gamma (1+|\mathbf{x}|^2)$ and if $\gamma>1/8m$ there are infinitely  many bound states in the spectral gap.
On the contrary, as proved by Birman in \cite{Birman}, if 
\[
\lim_{R\to \infty} R \int_{R}^\infty |A^0(r,\theta,\varphi)| dr = 0
\]
where $r,\theta, \varphi$ are standard spherical coordinates, the number of bound states with eigenvalues in the spectral gap is finite, see also \cite{cojuhari}.
In particular, for $A^0$ which is in $L^{3}\cap L^{\frac{3}{2}}$,
the $L^1$ norm of $|A^0|^3$ provides an upper bound for the number of bound states in the spectral gap and, more precisely, the number of eigenvalues in the spectral gap 
for the case of external potential $(\lambda A^0,0)$
grows as $(\lambda \|A^{0}\|_3)^3$ for large $\lambda$ \cite{Klaus}.

We conclude this section with the following remark. When the Cook criterion holds both for $D$ and $D+\mathcal{K}$ by Theorem 3.2 and Theorem 3.7 in Section X.3.2 in \cite{Kato},
    and since the spectrum of $D$ coincides with its absolutely continuous spectrum and it is $\mathbb{R} \setminus (-m,m)$ it holds that the spectrum of $D+\mathcal{K}$ is purely absolutely continuous and coincides with $\mathbb{R} \setminus (-m,m)$.

    Hence, when the external potential is smooth, has compact spatial support and is sufficiently small, the hypotheses of Lemma \ref{lm:Cook-criterion-D} and Lemma \ref{lm:Cook-criterion-D+K} are satisfied; hence the Cook criterion holds for both $D$ and $D+\mathcal{K}$ and thus $D+\mathcal{K}$ has no bound states.
    This result is compatible with the observation stated above about the bounds of the number of bound states in terms of the  suitable norms of the vector potential \cite{Klaus}.

\subsection{Non-equilibrium steady states (NESS)}

As discussed above, there are cases where the external potential $A$ is sufficiently strong to allow the
appearance of eigenvalues in the spectral gap region $(-m,m)$ for $D+\mathcal{K}$.

Consider $\tilde\omega_t$ one of these states, which is thus time dependent, and its time behavior depends on the particular form of the switching of the potential.
If the limit $t\to\infty$ does not exist, 
to analyze it at large times after the switching on of the external potential it is  useful to compute its ergodic mean on large intervals of time 
\[
\tilde\omega(B) =  \lim_{T \to \infty } \frac{1}{T}\int_{0}^T\tilde\omega_t(B) dt\ .
\]
If return to equilibrium or stability of KMS state does not hold, the state one obtains in this way is a non-equilibrium steady state  (NESS) \cite{Ruelle}.

We now discuss how to obtain some of these states with certain assumptions, which follow from the discussion presented in the last part of the previous section.
We assume in particular that the spectrum of the operator $D+\mathcal{K}$ is formed by the absolute continuous part coinciding to the essential part of the spectrum and the point spectrum which for the moment corresponds to a single eigenvalue which falls in the spectral gap of $D$, which is thus disjoint from the absolutely continuous spectrum.

If the Cook criterion holds for elements in the intersection of the domain of $D$ and $P_{ac} \mathcal{H}$ where $P_{ac}$ is the spectral projector over the absolute continuous part of the spectrum of $D+\mathcal{K}$ we get a well defined long time limit of the quasi-free state whose two-point function is obtained from the operator
\[
T_1(t) =   (P_{ac}+P_s)e^{i t (D+\mathcal{K})}  e^{-it D}\frac{1}{1+e^{-\beta D}} e^{it D} e^{-i t (D+\mathcal{K})}(P_{ac}+P_s)
\]
where $P_{ac}$ is the projector on the absolutely continuous part of the spectrum while $P_s$ is the orthogonal projector on the subspace generated by the eigenvector $\varphi_s$ of eigenvalue $s$.
We have furthermore used the fact that under the stated hypotheses $P_{ac}+P_s=I$.
The generalized Möller operators
\be
\tilde\Omega_\pm=\mathrm{s}-\lim_{t\to \pm\infty}
U_t^* P_{ac}\ , 
\ee
where $U_t=e^{it(D+\mathcal{K})}e^{-it D}$ coincides with \eqref{eq:new-UTonH},
exist thanks to the validity of the Cook criterion 
(Theorem 3.7 in Chapter X of \cite{Kato})
furthermore, they are such that  (Theorem 3.2 in Chapter X of \cite{Kato})
\begin{align*}
T_1(t) 
&=   
P_{ac}U_t\frac{1}{1+e^{-\beta D}}
 U_t^* P_{ac}
+
P_{s}U_t\frac{1}{1+e^{-\beta D}}
U_t^* P_{s}
\\
&\quad+
P_{ac} U_t \frac{1}{1+e^{-\beta D}}
U_t^* P_{s}
+
P_{s} U_t \frac{1}{1+e^{-\beta D}}
 U_t^* P_{ac}
\end{align*}
in the large time limit if the M\"oller operators exist we have that the first contribution converges to 
\[
P_{ac}\frac{1}{1+e^{-\beta (D +\mathcal{K}) } }
 P_{ac}
\]
because 
$ f(D)\tilde\Omega_\pm P_{ac}= \tilde\Omega_\pm  f(D+\mathcal{K}) P_{ac}$.   
For the same reason, the last two contributions vanish.
The remaining contribution does not depend on time 
we end up with 
\[
\tilde{T}_{1} = P_{ac}  \frac{1}{1+e^{-\beta (D+\mathcal{K})}}   P_{ac} 
+
P_s \frac{1}{1+e^{-\beta D}}  P_s\ .
\] 
The quasi-free state $\tilde{\omega}_1$ which has two-point functions constructed from the $T_{1}$ function
\[
\tilde\omega_1(\psi(f)\psi(g)^*) = \langle f, \tilde{T}_1g\rangle
\]
is in invariant under time evolution, however, since it is not an equilibrium state, it is a NESS. 
Notice that since the large time limit is well defined,  
considering the ergodic mean does not alter this observation.

If now the discrete eigenvalues in the spectral gap are $\{s_1,\dots, s_N\}$ with $N>1$, assuming as before that the corresponding eigenspaces have dimension $1$,  that they are disjoint from the absolutely continuous spectrum and that $P_{ac}+\sum_{i=1}^NP_{s_i}=I$, we do not get a definite limit for long time for
\[
T_N(t) = e^{i t (D+\mathcal{K})}\frac{1}{1+e^{-\beta D}} e^{-i t (D+\mathcal{K})}\ .
\]
Actually, the part of $T_N$ projected on the point spectrum is
\[
T_{N,pp}(t) = \sum_{j,k} P_{s_j} U_t \frac{1}{1+e^{-\beta D}} U_t^* P_{s_k} = \sum_{j,k} e^{i(s_j-s_k)t} 
P_{s_j} \frac{1}{1+e^{-\beta D}}  P_{s_k}\ ,
\]
and there is no a priori reason for having vanishing $P_{s_j} (1+e^{-\beta D}) P_{s_k}$ for different $j,k$.
The part of $T_N(t)$ projected on the absolutely continuous part behaves in the same way as for $T_1(t)$. 

We can nevertheless obtain a non-equilibrium steady state taking an ergodic mean on large times instead of the simple long time limit. The two-point function of this state is actually obtained from the operator 
\begin{equation}\label{eq:TN}
\tilde{T}_N=\lim_{T\to \infty}\frac{1}{T}\int_0^T T_{N}(t) dt
=
P_{ac}  \frac{1}{1+e^{-\beta (D+\mathcal{K})}}   P_{ac} 
+
\sum_j P_{s_j} \frac{1}{1+e^{-\beta D}}  P_{s_j}\ .
\end{equation}
If we use this operator $\tilde{T}_N$ to construct the two-point function of another quasi-free  state $\tilde{\omega}_N$
\begin{equation}\label{eq:tildeomegaN}
\tilde\omega_N(\psi(f)\psi(g)^*) \doteq  \langle f, \tilde{T}_Ng\rangle
\end{equation}
we obtain a new state which is invariant under time evolution and in general not of equilibrium. It is thus a NESS.

\section{Relative entropy and entropy production}
\label{eq:relative-entropy}

Consider two states $\Psi$ and $\Phi$ described by  positive trace class operators $\rho_\Psi$ and $\rho_\Phi$ on some Hilbert space.
The relative entropy of $\Phi$ with respect to $\Psi$ is given by
\[
\mathscr{S}(\Psi,\Phi) = \Tr (\rho_{\Psi}\left(\log(\rho_\Psi)-\log(\rho_\Phi)\right)).
\]
This definition of relative entropy
has been generalized by Araki in \cite{Araki, Araki2} to the case of normal states on a von Neumann algebra $\mathfrak{A}$ (see also the book of Ohya and Petz \cite{OhyaPetz} for the use and relevance of this concept).
Araki's relative entropy of a normal state $\Phi$ with respect to another normal state $\Psi$ of a $C^*$-algebra $\mathfrak{A}$ is defined as
\[
\mathscr{S}_{Araki}(\Psi,\Phi)= -\langle \Psi, \log(\Delta_{\Phi \Psi})\Psi \rangle 
\]
here, $\Delta_{\Phi\Psi}$  is the relative modular operator obtained as $\Delta_{\Phi\Psi}=S^*S$ where $S$ is the closed operator that implements at the same time the conjugation and the change of the reference state, $SA\Psi = A^*\Phi$, $A\in\mathfrak{A}$.

Let $(\mathfrak{A},\tau_t)$ be some $C^*$ dynamical system. Let $\omega^{\beta}$ be a KMS state with respect to the time evolution $\tau$ and $\omega^{\beta K}$ the KMS state of a perturbed time evolution constructed as in
\eqref{eq:araki-state} for some selfadjoint $K\in\mathfrak A$. $\omega^{\beta K}$ is represented by a vector in the GNS-representation of $\omega^\beta$, and the perturbed time evolution differs by an inner cocycle $(U_t)$ from the original one. 

Under these hypotheses, Araki's relative entropy takes a particularly simple representation given in terms of $K$ and $U$ only, 
\[
\mathscr{S}_{0,K}\doteq\mathscr{S}(\omega^\beta,\omega^{\beta K})=\omega^{\beta}(\beta K)+\log(\omega^{\beta}(U(i\beta))).
\]
See Proposition 1 in \cite{Petz}, section 6.2.3 in the book of \cite{BratteliRobinson} and chapter 12 in  \cite{OhyaPetz} where slightly different sign conventions are considered.
Notice that, in this case, the relative modular operator coincides with $U(i\frac{\beta}{2}) \Delta_\Psi U(i\frac{\beta}{2})^*$, up to a normalization factor,
where $\Delta_\Psi$ is the modular operator of the unperturbed state
and the normalization factor is $\omega^{\beta}(U(i\beta))$, see \cite{Petz, BratteliRobinson}.
This representation of the relative entropy makes manifest the thermodynamic interpretation of this formula. Actually $\omega^\beta(K)$ is the mean value of the relative Hamiltonian and $\beta^{-1}\log (\omega^\beta(U(i \beta)))$ is interpreted as the relative free energy of the system. Furthermore, reverting the point of view,
the equilibrium state for the perturbed dynamics can be seen as the minimizer of the relative free energy $\omega^\beta(K)-\beta^{-1} \mathscr{S}_{0,K}$.
See chapter 12 on perturbation of states in the book of Ohya and Petz \cite{OhyaPetz} for further details.
That book also contains a discussion on the use of similar expressions for the free energy and for the relative entropy when $K$ is an element of a larger set $\mathfrak{A}^{ext}$ that contains operators affiliated with $\mathfrak{A}$, namely lower bounded operators whose spectral projections are elements of $\mathfrak{A}$.

Actually, in order to apply that formula it suffices to have control on $K$ and the cocycle it generates.  Then it can be used also in other contexts, e.g. 
in the case of equilibrium states constructed with perturbation theory  \cite{DFP2}. 

We now discuss how to obtain $\mathscr{S}_{0,K}$ for the setup discussed in the present paper. In Proposition \ref{prop:PowerStormer} in the Appendix it is proved that $\omega^\beta$ and $\omega^{\beta K}$ are quasi-equivalent. 
We proceed now to estimate the relative entropy of $\omega^{\beta K}$ relative to $\omega^\beta$.

We show in particular that $\mathscr{S}_{0,K}$ can be given with operations that involve operators acting in the one-particle Hilbert space. 

We recall that (see, e.g., Appendix B in \cite{DFP2} or the PhD Thesis of Falk Lindner \cite{Lindner})
the expectation value of the logarithm of the normalization factor admits a decomposition similar to \eqref{eq:omegauib}, actually
\[
\log(\omega^{\beta}(U(i\beta))) =    \sum_{n>0} (-1)^n \int_{0<u_1\dots <u_n<\beta} dU \omega_c^{\beta}( \tau_{i u_1}(K), \dots ,\tau_{i u_n}(K))
\]
where $\omega^{\beta}_c(\cdot ; \dots ; \cdot)$ are the truncated or connected correlation functions of $\omega^{\beta}$.
Hence, in view of the time translation invariance of $\omega^\beta$ under $\tau_t$, we have
\[
\mathscr{S}_{0,K}=\sum_{n>1} (-1)^n \int_{0<u_1\dots <u_n<\beta} dU \omega_c^{\beta}( \tau_{i u_1}(K), \dots ,\tau_{i u_n}(K))\ .
\]
so the series starts with $n=2$.

The interaction Hamiltonian $K$ has
the form given in \eqref{eq:int-ham} $
K = -\sum_i (\psi( f_i) \psi(\mathcal{K}f_i)^*
+c_i)$
where $f_i$ is a basis of the one-particle Hilbert space $\mathcal{H}$ and where the constants $c_i$ depend on normal ordering. In the truncated functions of higher degree than 1 the constants do not contribute. Hence we can treat $K$ as a quadratic functional of the Dirac fields. 

Thanks to this form of $K$, the connected correlation functions among various factors $K$ are particularly simple,
and $\mathscr{S}_{0,K}$ admits a representation in terms of operations performed in the one-particle Hilbert space.
We find
$
\mathscr{S}_{0,K}= \Tr (\mathfrak{S})
$
where the operator $\mathfrak{S}$ that acts on the one-particle Hilbert space can be obtained in the following way.

Notice that, using the KMS condition for $\omega^\beta$, we have 
\[
\frac{d}{d u} \log \omega^\beta(U(i\beta u)) 
=
- \beta \frac{\omega^\beta (U(i\beta u) \tau_{i \beta u}K ) }{\omega^\beta (U(i\beta u))}
=
- \beta \frac{\omega^\beta (\tau_{i\beta(u-1)}(K) U(i\beta u)) }{\omega^\beta (U(i\beta u))}.
\]
Similarly to \eqref{eq:omegauib} and following a similar proof given in \cite{BKR} we have for $u\in (0,1)$ 
\begin{equation}\label{eq:Uibetau}
 \frac{\omega^\beta (B U(i\beta u)) }{\omega^\beta (U(i\beta u))}
= \omega^{\beta}(B)  +   \sum_{n} (-1)^n\int_{0<u_1\dots <u_n<\beta u} dU \omega_c^{\beta}(B, \tau_{i u_1}(K), \dots ,\tau_{i u_n}(K))
\end{equation}
the latter formula can be written in terms of elements of the one-particle 
Hilbert space as we have done to obtain 
\eqref{eq:expomegabetav1} and 
\eqref{eq:expomegabetav}. 
Actually,  using \eqref{eq:Uibetau} with $K$ given in \eqref{eq:int-ham} and with $B\doteq   \tau_{i\beta(u-i)}K$
we
can get the truncated correlation function in the $n$-th contribution in the sum present at the right hand side in \eqref{eq:Uibetau} as the sum over all possible connected graphs joining the $n+1$ vertices where the vertices correspond to the $n+1$ arguments of $\omega_c^\beta$  and the edges to the Fermi factors $F_{\pm}$.
Since the number of edges for every vertex is in this case $2$ and the number of edges equals the number of vertices the form of these graphs is well under control and it is equal to the number of possible permutations of $n+1$ elements keeping fixed the first one.
Arguing as we have done in \eqref{eq:expomegabetav1} to obtain 
\eqref{eq:expomegabetav}, the contribution of all the different graphs can be combined in a single integral whose domain is extended from the simplex $\beta \mathcal{S}$ to the cube $(0,\beta u)^n$.

Computing in this way $\frac{d}{du} \log \omega^\beta(U(i\beta u))$, after integrating over $u\in (0,1)$ and recalling that 
\begin{align*}
\mathscr{S}_{0,K}&=
\omega^\beta(\beta K) + \log \omega^\beta(U(i\beta )) 
\end{align*}
we actually get the following 
\begin{equation}\label{eq:rel:entropy-trace}
\mathscr{S}_{0,K} = \Tr \mathfrak{S}
\end{equation}
where 
\begin{align} 
\label{eq:S-expansion}
\mathfrak{S} 
&\doteq-
\beta\int_0^1 du\;
 \sum_{n>0}
   (-1)^n\!\!\!\!\!\int\limits_{U\in (0, \beta u)^n}\!\!\!\!\! dU 
    \; 
    \sqrt{F( \beta u-u_n)}\mathcal{K}F(\beta - \beta u+u_1) 
    \prod_{j=1}^{n-1} \left(\mathcal{K} F(u_{j+1}-u_j) \right) 
    \mathcal{K} \sqrt{F( \beta u-u_n)}
     \ .
\end{align}
The square roots of the operator $F(\beta u-u_n)$ are present in $\mathfrak{S}$ in that place because we have inserted the resolution of unit $ |f_i\rangle\langle f_i|$  to obtain the trace of the operator splitting one of the propagators in two. This is done 
to make the operator $\mathfrak{S}$ trace class when $\mathcal{K}$ is sufficiently small. We shall prove this claim below. A similar argument has been used in \eqref{eq:formula-omega_truncated} in the appendix. 

We can now present the following Lemmata proving that the operator $\mathfrak{S}$ is well behaved if $\mathcal{K}$ is sufficiently small.

\begin{lm}\label{le:Sbounded}
Let ${A}$ be smooth, of compact spatial support, constant for $t>0$ and vanishing for $t<-\epsilon$. Suppose that the corresponding operator $\mathcal{K}$ constructed as in \eqref{eq:new-defK} is such that
$\|\mathcal{K}\| <\beta^{-1}$.  
Then the operator $\mathfrak{S}$ on $\mathcal{H}$ given in 
\eqref{eq:S-expansion}
is bounded.
\end{lm}
\begin{proof}
$\mathcal{K}$ is bounded by hypothesis. 
For every $u,u_j$ in the domain of integration all the $F$ are bounded and $F(\beta u-u_n)$ is positive.
We thus get the estimate
\[
\|\mathfrak{S}\| \leq \frac{ \beta^2 \|\mathcal{K}\|^2 }{1-\beta \|\mathcal{K}\|}.
\]
\end{proof}

\begin{lm}\label{lm:S-traceclass}
    Consider $A_\mu(t,\bullet)$ such that it is smooth, bounded, of compact spatial support, it vanishes for $t<-\epsilon$ and it is constant for $t>0$.
    Suppose that the corresponding $\mathcal{K}$ is such that its integral kernel in the Fourier domain satisfies the bound
    \[
    |\hat{\mathcal{K}}(p,q)| \leq H(p-q)
    \]
    with
    $H\in L^1(\mathbb{R}^3)\cap L^2(\mathbb{R}^3)$
   and where   $\|H\|_1 < \beta^{-1}$, $\|H\|_2 < \beta^{-1}$ as in \eqref{eq:defH}. Then the operator 
    $\mathfrak{S}$ is trace class.
\end{lm}
\begin{proof}
We start observing that by the Young convolution inequality $\|\mathcal{K}\|\leq \|H\|_1$, hence $\|\mathcal{K}\| <  \beta^{-1}$.

We consider $\mathfrak{S}=\sum_{n> 0} \mathfrak{S}_n$
where the sum is over the number of operators $\mathcal{K}$ appearing in $\mathfrak{S}$ (minus $1$).

After some standard manipulation of the terms in the sum defining $\mathfrak{S}$ in \eqref{eq:S-expansion} and the change of the integration variables $u_j\to u-u_j$, we obtain  
\begin{align}\label{eq:def-Sn}
\mathfrak{S}_n 
&=
- \int_0^\beta du\;
   (-1)^{n}\!\!\!\!\!\int\limits_{U\in (0, u )^n}\!\!\!\!\! dU 
    \;
      \sqrt{\frac{e^{-u_nD }}{1+e^{-\beta D}}  }
      \mathcal{K}
    \frac{e^{-(\beta - u_1)D }}{1+e^{-\beta D}}
  \left(
  \prod_{j=1}^{n-1}
  \mathcal{K}
  \frac{(-1)^{-\sigma(u_j-u_{j+1})}e^{-(u_j-u_{j+1})D}}
  {1+e^{-\sigma(u_j-u_{j+1})\beta D} }
  \right)
   \mathcal{K}
      \sqrt{\frac{e^{-u_nD }}{1+e^{-\beta D}}} 
\end{align}
where $\sigma(u)$ denotes the sign of $u$.
Inserting the projectors $P+Q$ in front of each $\mathcal{K}$ and at the end of the integrand
 we have a decomposition of the integrand of $\mathfrak{S}_n$ into a sum of $2^{n+2}$ terms. 
Each term is characterized by a sequence of $n+1$ elements of the set $C=\{ P\mathcal{K}P,P\mathcal{K}Q, Q\mathcal{K}P, Q\mathcal{K}Q\}$.

We show that each term is a product of two Hilbert-Schmidt operators with bounded operators with bounded norms which are uniform in the $u$ and $U$ integration variables.

For terms which contain at least two of the factors with mixed $P,Q$ namely $Q\mathcal{K}P$ and its adjoint, this follows from Lemma \ref{lm:KHS}. 
Notice that once the place where $Q\mathcal{K}P$ and its adjoint appears in the product are fixed we can sum over all possible other choices of $P$ and $Q$. 
The Hilbert-Schmidt norm of this sum is then bounded by $\|\mathcal{K}\|^{n+1} \| P \mathcal{K} Q \|_{HS}\| Q \mathcal{K} P \|_{HS}$ uniformly on the integration domain.   There are at most $(n+1)n$ bounds of this form corresponding to the number of different places in $\mathfrak{S}_n$ where the factor $P\mathcal{K}Q$ and $Q\mathcal{K}P$ can appear.

For terms containing only one such factor, (say  $Q\mathcal{K}P$, the other follows similarly), we have three cases: 
1) if this factors appears in the first factor
notice that for every element of the domain of $u$ and $U$
at least one element $w\in 
Y=\{\beta-u_1, u_{1}-u_{2}, \dots , u_{n-1}-u_n,u_n\}$ is positive and bigger than 
$\beta/(n+1)$. (This holds even if generically on the domain of $U$, $u_l-u_{l+1}$ for some $l$ can be negative.)
But then another factor $\mathcal{K}$ is multiplied by $P\sqrt{\frac{e^{-{w}D}}{1+e^{-\beta D}}}$ which yields another Hilbert Schmidt operator with Hilbert-Schmidt norm bounded by $||H||_2||e^{-\beta/(2(n+1))\omega(p)}||_2$. 

2) An analogous result holds in the case that the last factor is $Q\mathcal{K}P$, observing that  
 for every element of the domain of $u$ and $U$
at least one element 
$\tilde{w}\in 
\tilde{Y}=\{u_1, u_{2}-u_{1}, \dots , u_n-u_{n-1},\beta - u_n\}$ is bigger than $\beta/(n+1)$, hence another factor $\mathcal{K}$ is multiplied by $Q\sqrt{\frac{e^{{\tilde{w}}D}}{1+e^{\beta D}}}$ whose Hilbert-Schmidt norm is then again bounded by
$||H||_2||e^{-\beta/2(n+1)\omega(p)}||_2$.

3) If the single factor $Q\mathcal{K}P$ is present in middle we have, for every element of the integration domain of $u$ and $U$, that either $u_n \geq \beta/2$ or $\beta - u_n \geq \beta/2$. The second factor which is Hilbert Schmidt is then either the first $\mathcal{K}$ multiplied by $Q \sqrt{\frac{e^{(\beta -u_n)D}}{1+e^{\beta D}}}$ or the last one 
multiplied by $P \sqrt{\frac{e^{-u_n D}}{1+e^{-\beta D}}}$. In both cases, the Hilbert-Schmidt norm of this second factor can then bounded by $||H||_2||e^{-\beta/(2(n+1))\omega(p)}||_2$. 

There are hence at most $2(n+1)$ terms in the decomposition of the integrand that can be bound in this way.

In case no such factor occurs, one even obtains two factors $P\mathcal K P$ or two factors $Q \mathcal K Q$  multiplied from the right or the left by an exponentially decaying function of momentum of the form 
$P\sqrt{\frac{e^{-{w}D}}{1+e^{-\beta D}}}$ or 
$Q\sqrt{\frac{e^{\tilde{w}D}}{1+e^{\beta D}}}$ with $w \in Y$ and $\tilde{w} \in \tilde{Y}$ both bigger than $\beta/(n+1)$.
Therefore again a product of two Hilbert Schmidt operators with bounded operators and the norm respects the bound $||H||_2||e^{-\beta/(2(n+1))\omega(p)}||_2$ uniformly on every $u$ and $U$ in their integration domain.
There are hence at most two terms in the decomposition of the integrand that can be bound in this way.

Combining all these estimates, which are uniform on the $u$ and $U$ integration domain, taking the integrals we have
\begin{equation}\label{eq:step70}
\begin{aligned}
\Tr(|\mathfrak{S}_n|)  &\leq \beta^{n+1} 
\|
\mathcal{K}\|^{n-1}
\left(  n(n+1) 
\|Q \mathcal{K} P\|_{HS} \|P \mathcal{K} Q\|_{HS} 
+
(n+1) \|Q \mathcal{K} P\|_{HS} I_n
+
(n+1) \|P \mathcal{K} Q\|_{HS} I_n
+ 2 I_n^2
\right).
\\
&\leq \beta^{n+1} 
\|
\mathcal{K}\|^{n-1}
\left(  n(n+1) C^2 
+
(n+1)C I_n
+
(n+1)C I_n
+ 2 I^2_n
\right)
\end{aligned}
\end{equation}
where $C$ bounds $\|P\mathcal{K}Q \|_{HS}$ and  $\|Q\mathcal{K}P \|_{HS}$ as in Lemma \ref{lm:KHS}
and 
$I_n$ is such that
\[
||H||_2||e^{-\beta/(2(n+1))\omega(p)}||_2  \leq ||H||_2 \sqrt{\frac{8\pi}{\beta^{3}} (n+1)^3} \doteq I_n. 
\]
Hence the sum over all $n$ can  be taken under the hypothesis of the present lemma and it gives a finite result. 
\end{proof}

\bigskip

We combine now the results of Lemma  \ref{lm:S-traceclass} and Lemma \ref{lm:KHSunifrom} to discuss conditions under which
the relative entropy vanishes in some limits.
Consider an external potential $A$ which is sufficiently small (smooth and of compact spatial support and time independent) so that
return to equilibrium holds for $\A(t)=h_{\mathcal{T}}(t)\A$, where $h_{\mathcal{T}}(t)$ is a switch on function of the form discussed in Lemma \ref{lm:KHSunifrom}.
In other words, we assume that 
 for every $\mathcal{T}>0$, considering $\mathcal{K}$ obtained from $\A(t)$ above, the corresponding
 $D+\mathcal{K}$
has purely absolutely continuous spectrum and the Cook criterion holds for $D+\mathcal{K}$ thanks to Lemma \ref{lm:Cook-criterion-D+K}. 
Under these hypotheses we have that $\|P\mathcal{K}Q\|_{HS}$ tends to $0$ in the limit $\mathcal{T}\to\infty$.
We would like to prove that also the relative entropy in the form discussed in this paper vanishes under that limit. 

To prove it we might operate as in the proof of Lemma \ref{lm:S-traceclass}
where we observed that $\mathfrak{S}$ is at least quadratic in $\mathcal{K}$ and that the Fermi factors inserted in the expansion of $\mathscr{S}_{0,K}$ essentially 
act as $P$ or $Q$.

In view of the decomposition of 
$\mathfrak{S} = \sum_{n>0} \mathfrak{S}_n$ discussed in the proof of Lemma 
 \ref{lm:S-traceclass} with 
 $\mathfrak{S}_n$ given in \eqref{eq:def-Sn},  
 and in view of the inequality \eqref{eq:step70} obtained therein, we have to estimate the decay in $\mathcal{T}$ of $I_n$ for large $\mathcal{T}$.
 The most severe contribution in the estimate of the decay for large $\mathcal{T}$ in $\Tr{\mathfrak{S}_n}$ comes from the evaluation of the lowest order. The higher orders can be treated as in Lemma \ref{lm:KHSunifrom}.
We therefore need to study the $L^2$ norms of terms of the form
\[
P\mathcal{K}_{0,0}P(p,q)e^{-\frac{\beta}{2} \omega(p)} \ .
\]

This is essentially equal to the $L^2$ norm of
\[
O=\hat{\dot{h}}(\mathcal{T}(\omega(p)-\omega(q))) \hat{\A}(p-q)e^{-\frac{\beta}{2} \omega(p)} 
\]
where the integral over $t$ on $\mathcal{K}_{0,0}$ corresponds to the Fourier transform of $\dot{h}_{\mathcal{T}}$ evaluated at $\omega(p)-\omega(q)$, and we have related it to the Fourier transform of $\dot{h}$.
We thus have
\[
\|O\|^2_{2} \leq 
\int d^3p d^3q |\hat{\dot{h}}(\mathcal{T}(\omega(p)-\omega(q)))|^2 
|\hat{\A}(p-q)|^2e^{-{\beta} \omega(p)}\ . 
\]
We divide  the domain of $p, q$ integration in the following way, 
\[
\int d^3 p\int_{|(p-q)(p+q)|\leq \epsilon^2} d^3(p-q)|O|^2
+
\int d^3 p\int_{|(p-q)(p+q)|\geq \epsilon^2} d^3(p-q)|O|^2.
\]
The integral of the first term over $p-q$ is bounded by $c\,\epsilon  p^2e^{-\beta\omega(p)}$ with a constant $c$ which does not depend on $\epsilon$, so its contribution can be made arbitrarily small. For the second term one exploits the fact that $|\omega(p)-\omega(q)|>\epsilon^2/2m$, hence the integrand vanishes pointwise in the limit $\mathcal{T}\to\infty$. Since it is uniformly bounded by a an integrable function, the integral vanishes by dominated convergence.
A similar analysis can be performed also for the $n,l$ contribution in $\mathcal{K}_{n,l}$.

\bigskip
We conclude this section with a discussion about the use of formula \eqref{eq:S-expansion} and \eqref{eq:rel:entropy-trace} to compute the relative entropy in another context. We study the relative entropy of $\tilde{\omega}_N$ given in \eqref{eq:tildeomegaN} with respect to $\omega^{\beta K}$ given in \eqref{eq:araki-state}. Here we change reference states for the relative entropy to obtain closed explicit expressions. In order to compute this relative entropy we compare the corresponding one-particle Hilbert space operators $T(\beta)$ of \eqref{eq:Tbeta} and $\tilde{T}_N$ given in \eqref{eq:TN} and we use \eqref{eq:rel:entropy-trace} with 
$D+\mathcal{K}$ at the place of $D$ and 
$\sum_{j} (d_{s_j}-s_j) P_{s_j} = k_j P_{s_j}$ at the place of $\mathcal{K}$, 
where $d_j\in \mathbb{R}$ is such that  
$P_{s_j}\frac{1}{1+e^{-\beta D}} P_{s_j}=\frac{1}{1+e^{-\beta d_j}}
P_{s_j}$. 
Observe that $d_j$ exists because $1/(1+e^{-\beta D})$ is positive and the image of $1/(1+e^{-\beta d_j})$ with varying $d_j$ is $\mathbb{R}^{+}$.
Furthermore, with these choices 
\[
\tilde{T}_N= 
P_{ac}
\frac{1}{1+e^{-\beta (D+\mathcal{K} )}}
P_{ac}
+
\sum_j
P_{s_j}
\frac{1}{1+e^{-\beta D}}
P_{s_j}
=
\frac{1}{1+e^{-\beta (D+\mathcal{K} + \sum_j k_j P_{s_j} )}}.
\]
We start analyzing the case $N=1$ and operating the substitutions described above in \eqref{eq:rel:entropy-trace} and in \eqref{eq:S-expansion},  with an obvious extension of the notation, we obtain 
\begin{align*}
\mathscr{S}(\omega^{\beta K},\tilde\omega_1) 
&= -\beta 
k_1
\;
 \sum_{n>0}
   \frac{ (-1)^n(k_1)^n e^{-\beta s_1}}{(1+e^{-\beta s_1})^{n+1}} \frac{1}{n+1}\!\!\!\!\!\int\limits_{U\in (0,\beta )^n}\!\!\!\!\! dU 
    \;
  \left(
  \prod_{j=1}^{n-1}
  (\theta(u_j-u_{j+1})-\theta(u_{j+1}-u_j)e^{-\beta s_1}  ) 
  \right).
\end{align*}
To compute the integrals in $\mathscr{S}(\omega^{\beta K},\tilde{\omega}_1)$, we use the expansion of the integral in $U$ as a sum over permutations, as in   
\eqref{eq:expomegabetav1}
which is equivalent to  \eqref{eq:expomegabetav}.
We obtain 
\begin{align*}
\mathscr{S}(\omega^{\beta K},\tilde\omega_1) 
& =
   e^{\beta s_1}\sum_{n>0}
   \left(
   \frac{ \beta k_1 }{1+e^{\beta s_1}} 
   \right)^{n+1}
    \frac{1}{(n+1)!} 
    \sum_{l=0}^{n-1}  
    (-1)^{l} 
    e^{\beta l s_1}
     A(n, l) 
 \\
&=
   e^{\beta s_1} 
   \sum_{n>0}
   \left(\frac{ k_1 \beta 
    }{1+e^{\beta s_1}}
    \right)^{n+1}
    \frac{1}{(n+1)!} 
    A_n(-e^{\beta s_1})
\\
     &=
    \frac{ \beta k_1 
    }{1+e^{\beta s_1}}+
    \log\left(\frac{1 + e^{-\beta (s_1+k_1)}}{1 + e^{-\beta s_1}}\right)
\end{align*}
where $A(n,l)$ are called Eulerian numbers and count the number of permutations of the first natural number $n$ with $l$ positive increments (raises). 
We furthermore denoted by $A_n(u) = \sum_{k=0}^{n-1} A(n,k) u^n$ the  Eulerian polynomial which are known to admit the following generating function \cite{AdvancedCombinatorics}
\begin{align*}
G(x,u) 
&=\sum_{n\geq 1}{A_n(u)} \frac{x^{n}}{n!} = \frac{u-1}{u-  e^{x(u-1)}}-1
\end{align*}
which is used in the integrated form above 
\begin{align*}
\tilde{G}(x,u) 
&= \int_{0}^x G(y,u)\;dy  
= 
\sum_{n\geq 1}{A_n(u)} \frac{x^{n+1}}{(n+1)!}=
- x + 
\frac{1}{u}\log\left( \frac{1-u}{1-u e^{x(1-u) }} \right) .
\end{align*}
By direct inspection, we get that the expression of the relative entropy obtained above is positive, it admits a Taylor expansion for $k_1=0$ whose lower order contribution in $k_1$ is $k_1^2$.
This formula has a direct physical interpretation. The first term proportional to $k_1$ in the entropy is the statistical mean of the variation of internal energy in $\tilde\omega_1$ wrt to $\omega^{\beta V}$ and the logarithmic term is minus the difference of the free energy of the eigenstate $s_1+k_1$ wrt the eigenstate of $s_1$.
More precisely let $Z(\beta, \epsilon) = 1+e^{-\beta (s_1+\epsilon k_1)}$ be the partition function of the one-particle fermionic state of energy $s_1+\epsilon k_1$ we have that
\[
\mathscr{S}(\omega^{\beta K},\tilde\omega_1) = -\frac{\partial}{\partial\epsilon} \log Z (\beta, \epsilon) |_{\epsilon=0}+ \log{Z(\beta, 1)}-\log{Z(\beta, 0)}.
\]
To get the relative entropy in the case of $\tilde\omega_N$ we just need to sum over all various eigenvalues, we obtain
\begin{equation}\label{eq:Sexample}
\mathscr{S}(\omega^{\beta K},\tilde\omega_N) 
=\sum_{j=1}^N
    \left(
    \frac{ \beta  k_j 
    }{1+e^{\beta s_j}}+ 
    \log\left(\frac{1 + e^{-\beta (s_j+k_j)}}{1 + e^{-\beta s_j}}\right)
    \right)
\end{equation}
which has the same properties of $\mathscr{S}(\omega^{\beta K},\tilde\omega_1) $ and the same thermodynamical interpretation of $\mathscr{S}(\omega^{\beta K},\tilde\omega_1) $ presented above with partition function $Z(\beta, \epsilon) = \prod_j (1+e^{-\beta (s_j+\epsilon k_j)})$.

In these concrete examples we used $\omega^{\beta K}$ as reference state instead of $\omega^{\beta}$ to analyze $\tilde{\omega}_N$ in order to get closed expressions that can be interpreted as descending from the partition function.

\subsection{Relative entropy for generic quasi-free states}

Motivated by the concrete example presented at the end of the previous section, we discuss now how to write the relative entropy in \eqref{eq:S-expansion} in a more compact form with a direct connection to representations of the relative entropy already used in the literature.
We observe in particular that the relative entropy of $\omega^{\beta K}$ relative to $\omega^\beta$ can be computed as the trace of the operator given in \eqref{eq:S-expansion}. 
We have furthermore that the integral over $u$ in \eqref{eq:S-expansion} can be computed and, recalling also
\eqref{eq:expomegabetav1}, the trace of \eqref{eq:S-expansion}
can thus be written as
\begin{align*}
\mathscr{S}_{0,K} &= 
-\Tr \left(
\sum_{n>0}
 \sum_{\pi\in\mathcal{P}_n}
   (-1)^n\!\!\!\!\!\int\limits_{0<u_1<\dots <u_n<\beta}\!\!\!\!\! dU 
    u_1
    \;
    \sqrt{F(\beta -u_{\pi_n})}
    \mathcal{K} 
   F(u_{\pi_1}) 
  \left(
  \prod_{j=1}^{n-1}
  \mathcal{K}
  F(u_{\pi_{j+1}}-u_{\pi_j}) 
  \right)
   \mathcal{K}
  \sqrt{F(\beta -u_{\pi_n})}  \right)\ .
\end{align*}
Exploiting now the cyclic property of the trace and rearranging the domains of integration we can change the factor $u_1$ in the integrand to  $(u_{j+1}-u_{j})$ for any $j\in\{1,\dots n\}$ (with $u_{n+1}= \beta$).
Hence, averaging all these contributions we obtain a factor $\beta/(n+1)$,
\begin{equation}\label{eq:new-S0V}
\mathscr{S}_{0,K}= 
-\Tr \left(
\beta
 \sum_{n>0}
   \frac{(-1)^n}{n+1}\!\!\!\!\!\int\limits_{U\in (0,\beta )^n}\!\!\!\!\! dU 
      \sqrt{F(\beta -u_n)}
    \mathcal{K}
   F(u_1) 
  \left(
  \prod_{j=1}^{n-1}
  \mathcal{K}
  F (u_{j+1}-u_j) 
  \right)
   \mathcal{K}
  \sqrt{F(\beta -u_n)}  
\right) \ .
\end{equation}
The latter expression can be written in a compact form.

To this end, 
similar to the operator valued function $F(u)$ given in  \eqref{eq:fermi-factors} we consider the operator valued function defined for $u$ in $[0,\beta)$
as
\[
F_K(u)\doteq \frac{e^{- u (D+\mathcal{K})}}{1+e^{-\beta (D+\mathcal{K})}}.
\]
and extended by anti-periodicity to $u\in \mathbb{R}$, namely there $F_K(u) =(-1)^{[\frac{u}{\beta}]} F_K(u-\beta[\frac{u}{\beta}])$ where $[u]$ denotes the largest integer smaller than or equal to $u$. This function has values in the set of bounded operators on the one-particle  Hilbert space $\mathcal{H}$.

As proved in Proposition \ref{prop:derivatives} this function for various $\mathcal{K}$, satisfies an expansion in powers of $\lambda$ which can be used to relate $F_{\lambda K}(u)$ to $F_0(u) = F(u)$.
That expansion can be shown to be convergent by estimating the reminder of the truncated expansion,
as we have done in Theorem \ref{lm:T} and equation \eqref{eq:recursive-relation-simplified} in particular.
We have actually the following

\begin{proposition}\label{prop:newS}
Suppose that $\|\mathcal{K}\| < \beta^{-1}$,
the relative entropy given in \eqref{eq:new-S0V} is equal to
\begin{equation}\label{eq:new-mathfrakS}
\mathscr{S}_{0,K} = \Tr (\tilde{\mathfrak{S}})
\end{equation}
where
\begin{equation}\label{eq:new-SigmaTilde}
\tilde{\mathfrak{S}}
\doteq
\beta
\int_0^{\beta} du 
  \int_0^1  d\lambda \lambda
\sqrt{F(\beta-u)}\mathcal{K}  F_{\lambda K}(u)\mathcal{K} \sqrt{F(\beta-u)}
\ .
\end{equation}
If the hypotheses of Lemma \ref{lm:S-traceclass} are satisfied, the trace of $\tilde{\mathfrak{S}}$ is finite.
\end{proposition}

\begin{proof}
We start considering the operator under the trace at the right hand side of equation \eqref{eq:new-S0V}.
The sum over $n$ in this operator is absolutely convergent. Furthermore, it converges to $\tilde{\mathfrak{S}}$.
To prove this claim, we preliminary see that the power series in $\mathcal{K}$ appearing in \eqref{eq:new-S0V} can be directly obtained using the results of Proposition \ref{prop:derivatives}.
To show that the series is absolutely convergent we may then operate as in the proof of Theorem \ref{lm:T}, actually, 
using recursively equation \eqref{eq:recursive-relation-simplified} with $\mathcal{K}$ replaced by $u \mathcal{K}$ we can prove that the series at the right hand side of \eqref{eq:new-S0V} truncated at order $N$ in powers of $\mathcal{K}$ is equal to 
$\tilde{\mathfrak{S}}$ up to an error term which vanishes in the limit $N\to\infty$. 
Notice in particular that the factor $1/(n+1)$ present in the $n$-th order term of the expansion of the right hand side of \eqref{eq:new-S0V} arises from the integration of the factor $u^n$ on the interval $[0,1]$. 
Under the hypothesis of Lemma \ref{lm:S-traceclass} we know that $\mathfrak{S}$ is trace class and its trace is equal to $\mathscr{S}_{0,K}$ hence $\Tr\tilde{\mathfrak{S}}= \Tr \mathfrak{S}$ and it is finite. 
\end{proof}

The expression of the relative entropy presented in Proposition \ref{prop:newS}
and in equation \eqref{eq:new-S0V}
can be used to find a formula for the relative entropy valid under certain hypothesis for quasi-free states. 

\begin{proposition}\label{prop:gen-rel-entropy}
Let $A$, $B$ be strictly positive self adjoint operators on the one-particle Hilbert space 
$\mathcal{H}$ such that also $1-A$ and $1-B$ are strictly positive, where strictly positive means that 0 is not an eigen value. 

Consider the corresponding quasi-free states $\omega_A$ and $\omega_B$ characterized by the two-point functions
\[
\omega_A(\psi(f)\psi^*(g)) = \langle f,A g \rangle, \qquad 
\omega_B(\psi(f)\psi^*(g)) = \langle f,B g \rangle, \qquad f,g\in \mathcal{H}.
\]
Let 
\[
A(u)\doteq 
A^{1-u} \left(1-A\right)^{u}, \qquad B(u)\doteq B^{1-u} \left(1-B\right)^{u}
\]
for $u\in[0,1)$ and for $u\in\mathbb{R}$, $A(u)=(-1)^{[u]}A(u-[u])$, $B(u)=(-1)^{[u]}B(u-[u])$ where $[u]$ denotes the largest integer smaller or equal to $u$.  
Consider the operators $\mathcal{A}$ and $\mathcal{B}$ on $\mathcal{H}_e={L}^{2}([0,1])\otimes {\mathcal{H}}$ obtained extending the operator valued functions $A(u)$ and $B(u)$ to $\mathcal{H}_e$ as convolution operators in $u$ (similar to $\mathcal{F}$ and $\mathcal{F}_K$ above).
Consider 
\[
C \doteq \log (A^{-1}-1) - \log(B^{-1}-1)
\]
and  its extension $\mathcal{C}$ to $\mathcal{H}_e$ as multiplicative operator in $u$.
If ${C}$ is bounded and if 
$\mathcal{C}\mathcal{A}$ is 
Hilbert-Schmidt on $\mathcal{H}_e$, then
the Araki relative entropy of $\omega_B$ relative to $\omega_A$ is 
\be\label{eq:rel-entropy-general}
\mathscr{S}(\omega_A,\omega_B) = \log\det\!{}_2 (1+ \mathcal{C}\mathcal{A}).
\ee
Furthermore, $\mathcal{A}$ and $\mathcal{B}$ have  densely defined inverses
\[
\mathcal{A}^{-1}=-\log(A^{-1}-1)+\frac{d}{du}\qquad \mathcal{B}^{-1}=-\log(B^{-1}-1)+\frac{d}{du}
\]
where the derivative is applied to functions with antiperiodic boundary values, $\mathcal A^{-1}\mathcal B -1$ is Hilbert-Schmidt and
$\mathcal{B}^{-1}\mathcal{A}+\mathcal{A}^{-1}\mathcal{B}-2$ is trace class. In terms of these operators an alternative expression for the relative entropy is
\be\label{eq:rel-entropy-general2}
\mathscr{S}(\omega_A,\omega_B) = \log\det\!{}_2 (\mathcal{B}^{-1}\mathcal{A}).
\ee
where  $\det_2$ is the regularized determinant on the extended one-particle Hilbert space $\mathcal{H}_e = L^{2}(0,1)\otimes \mathcal{H}$.  
We recall here that 
\[
\log\det\!{}_2(1+\mathcal{C} \mathcal{A}) 
=
\Tr_e \left( \log  (1+ \mathcal{C} \mathcal{A}) - \mathcal{C} \mathcal{A} \right) 
\]
where $\Tr_e$ is the trace on $\mathcal{H}_e$.
The sum with the relative entropy of $\omega_A$ relative to $\omega_B$ is
\be\label{eq:rel-entropy-sum}
\mathscr{S}(\omega_A,\omega_B) +\mathscr{S}(\omega_B,\omega_A)=-\Tr_e(\mathcal{B}^{-1}\mathcal{A}+\mathcal{A}^{-1}\mathcal{B}-2)=-\Tr_{ e}\, \mathcal{C}(\mathcal{A}-\mathcal{B})=-\Tr(C(A-B)) \ .
\ee 
\end{proposition}
\begin{proof}
Since $A$ and $1-A$ are strictly positive, there is a selfadjoint operator $D_A$ with $A=(1+e^{-D_A})^{-1}$. The same holds for $B$ and $D_B$. In terms of these operators $C = D_B-D_A$.
We start by considering the right hand side of equation \eqref{eq:new-S0V},
substituting $D$ by $D_A$, and $\mathcal{K}$ by $C$. This gives then the relative entropy $\mathscr{S}(\omega_A,\omega_B)$. 
We observe that up to a cyclic permutation 
that expression can be written in terms of $\mathcal{A}$ and $\mathcal{C}$ on the extended Hilbert space $\mathcal{H}_e$.
It is actually
\begin{align*}
\mathscr{S}(\omega_A,\omega_B)
&= \Tr_e \left(\sum_{n\geq 2} \frac{(-1)^{n+1}}{n} (\mathcal{C}\mathcal{A})^{n} \right)
\end{align*}
where now the trace is for operators on $\mathcal{H}_e$ and the latter is well defined because $\mathcal{C}\mathcal{A}$ is Hilbert-Schmidt by hypothesis. 
Taking the sum over $n$ and recalling the form of the regularized determinant we have
\begin{align*}
\mathscr{S}(\omega_A,\omega_B)
&= \Tr_e \left( \log \left(1+ \mathcal{C}\mathcal{A} \right) - \mathcal{C}\mathcal{A}\right)
\\
&= \log\det\!{}_2 \left( 1+ \mathcal{C}\mathcal{A}\right),
\end{align*}
thus getting  equation \eqref{eq:rel-entropy-general}.

The formulas for the inverses of $\mathcal{A}$ and $\mathcal{B}$ can be directly verified, and $\mathcal{A}^{-1}-\mathcal{B}^{-1}=-\mathcal{C}$. 

We observe in particular that a relation similar to \eqref{eq:recursive-relation-simplified} proved in Theorem \ref{lm:T}, holds for $A(u)$, $B(u)$ and $C$. This relation for the operators 
$\mathcal{A}$ and $\mathcal{B}$, takes the form
\[
\mathcal{B} - \mathcal{A} = -\mathcal{B} \mathcal{C} \mathcal{A}.
\]
With this observation, the equations \eqref{eq:rel-entropy-general2} and \eqref{eq:rel-entropy-sum} follow using invertibility of $\mathcal{A}$ and $\mathcal{B}$ as well as the 
property of the regularized determinant
\be
\det\!{}_2(T)\det\!{}_2(T^{-1})=e^{-\Tr(T+T^{-1}-2)}
\ee
for invertible $T$ with $T-1$ Hilbert-Schmidt. 
Notice also that the operator $\mathcal{C}(\mathcal{A}-\mathcal{B})$ is trace class because
$\mathcal{C}(\mathcal{A}-\mathcal{B})= \mathcal{C}\mathcal{B}\mathcal{C}\mathcal{A}$, $\mathcal{C}\mathcal{A}$ is Hilbert-Schmidt by hypothesis and the same holds for $\mathcal{C}\mathcal{B} = \mathcal{C}\mathcal{A}-\mathcal{C}\mathcal{B}\mathcal{C}\mathcal{A}$ because $\mathcal{B}$ and $\mathcal{C}$ are bounded. Moreover, its integral kernel with respect to $u$, 
\[\mathcal{C}(\mathcal{A}-\mathcal{B})(u,u')=C(A(u-u')-B(u-u'))\]
is continuous and depends only on the difference of the variables, hence
\[\Tr_e{\mathcal{C}(\mathcal{A}-\mathcal{B})}=\int_0^1 du \,\Tr C(A-B)\;.
\]
\end{proof}

We observe that, with the notation of Proposition \ref{prop:gen-rel-entropy},  denoting by ${A}_\lambda(u) = (e^{u (D_A+\lambda C)}+e^{-(1-u)(D_A+\lambda C)})^{-1}$ so that $A_0(0)=A$ and $A_1(0) = B$ and considering $\mathcal{A}_\lambda$ the corresponding operators on the extended Hilbert space $\mathcal{H}_e$, the formula for the relative entropy obtained in Proposition \ref{prop:gen-rel-entropy} can be further simplified, we actually have that
\begin{align*}
\mathscr{S}(\omega^A,\omega^B) 
= \Tr_e \int_0^1 d\lambda 
\sum_{n\geq 2} (-1)^{n+1}  \mathcal{C} (\lambda \mathcal{A}\mathcal{C} )^{n-1} \mathcal{A}
= \Tr_e \int_0^1 d\lambda \;
 \mathcal{C} (\mathcal{A}_\lambda - \mathcal{A})
= \Tr \int_0^1 d\lambda \;
 {C} ({A}_\lambda(0) - {A})\,.
\end{align*}

Proposition \ref{prop:gen-rel-entropy}  can be also used to make manifest the relation of the Araki relative entropy and the Kullback–Leibler divergences used elsewhere in the literature, see \cite{KL} and e.g. \cite{Finster} where a similar formula for quasifree states in a simpler situation is obtained.
Recalling Proposition \ref{prop:gen-rel-entropy}, we  observe that \eqref{eq:rel-entropy-general2} can be written in the form 
\begin{align*}
\mathscr{S}(\omega^A,\omega^B) 
&= \Tr_e \left( \log\left( \mathcal{B}^{-1}\mathcal{A}\right) -( \log (A^{-1}-1) - \log(B^{-1}-1)))\mathcal{A} \right)
\end{align*}
This expression is very close and reminds  the following formula for the relative entropy derived in the case of compact operators in \cite{Finster}.
\begin{align*}
\Tr\left(A(\log (A)-\log(B)) +  (1-A) (\log(1-A)-\log(1-B))\right) \ .
\end{align*}

A similar formula for the entanglement entropy in the case of compact operators can be found in \cite{Finster}, see also \cite{Casini2008}.
The previous expression which holds under some simplifying hypothesis (e.g. in the finite dimensional case), makes manifest the relation of the Araki relative entropy and the Kullback–Leibler divergences used elsewhere in the literature, see \cite{KL} and e.g. \cite{Finster}. For the non compact case, Proposition \ref{prop:gen-rel-entropy}
and equations \eqref{eq:rel-entropy-general} \eqref{eq:rel-entropy-general2} furnishes generalized expressions for the relative entropy which holds also in the infinite dimensional situation.

Finally we recall that the analysis of the form of modular operator out of properties of the two-point function has been recently studied in \cite{Frob}
for the case of Fermions in 
two dimensions see also \cite{Cadamuro}. In this respect, results of Proposition \ref{prop:gen-rel-entropy} can be used to obtain corrections to modular operators under perturbations. 

\subsection{Non-equilibrium states and entropy production}

The machinery discussed in the previous section permits to analyze also the relative entropy between the equilibrium state and the one obtained composing to the perturbed dynamics. 
This open the possibility of obtaining concrete examples in which states which are not in equilibrium can be studied by thermodynamic like quantities.

Properties of the relative entropy and their relation to the notion of entropy production for states which are not in equilibrium and or stationary but not in equilibrium 
have been studied extensively by Jak{\v s}i\'c and Pillet, see e.g.  \cite{JP01, JP02, JP03}. See also the earlier works \cite{Oj0, Oj1, Oj2, Spohn} and \cite{Ruelle2, Ruelle3}

We use the analysis presented in the previous section about relative entropy to derive an expression which evaluates the entropy production (of relative entropy).
The concepts obtained are in close contact with the works of Jak{\v s}i\'c and Pillet.

We start introducing the state evolved in time. The state we are considering is the equilibrium state for the free theory. The evolution we are considering is the interacting time evolution
\[
\omega^{\beta}_t = \omega^\beta\circ \tau_t^K\ .
\]
At the level of one-particle Hilbert space this corresponds to consider
\[
\frac{1}{1+e^{-\beta D(t)}} \doteq  e^{i t (D+\mathcal{K})} \frac{1}{1+e^{-\beta D}} e^{-i t (D+\mathcal{K})}
\]
where 
\[
D(t) = e^{i t (D+\mathcal{K})} D e^{-i t (D+\mathcal{K})} = D + \mathcal{K}- e^{i t (D+\mathcal{K})} \mathcal{K}e^{-i t (D+\mathcal{K})}  = D+ \mathcal{K} - \mathcal{K}(t)  = D + {L}_t\ .
\]
 Now, we may use Proposition \ref{prop:newS} with $L_t = \mathcal{K} - \mathcal{K}(t)$ at the place of $\mathcal{K}$ to compute the relative entropy  of $\omega^\beta_t$ relative to $\omega^\beta$, namely using equation  \eqref{eq:new-SigmaTilde}
\begin{equation}\label{eq:rel-entropy-omegat}
\mathscr{S}(\omega^\beta, \omega^{\beta}_t )
=\beta \Tr
\int_0^{\beta} du 
  \int_0^1  d\lambda \lambda
\sqrt{F(\beta-u)}L_t  F_{\lambda  L_t}(u)L_t \sqrt{F(\beta-u)}
\ .
\end{equation}
We furthermore observe that 
\begin{equation}\label{eq:entropy-differences}
\mathscr{S}(\omega^\beta, \omega^{\beta}_t )
=
\int_0^t 
\frac{d}{d\tilde{t}} 
\mathscr{S}(\omega^\beta, \omega^{\beta}_{\tilde{t}} )d\tilde{t}
+
\mathscr{S}(\omega^\beta, \omega^{\beta} )\ .
\end{equation}
Hence, in analogy to the results of Theorem 1.1 in \cite{JP01},  
the following quantity is thus understood as the entropy production in the state $\omega^{\beta}_t$ relative to $\omega^\beta$
\[
E_t\doteq \frac{d}{d{t}} 
\mathscr{S}(\omega^\beta, \omega^{\beta}_{{t}} )\ .
\]
To get a more manageable expression of this quantity, we proceed as follows.
Denoting by 
\[
\Phi_t = \frac{d}{dt} D(t) = \frac{d}{dt} L_t =  e^{-i t (D+\mathcal{K})} [D,\mathcal{K}] e^{i t (D+\mathcal{K})}\ ,
\]
we can write the entropy production as
\begin{align*}
E_t 
=& \beta \Tr
\int_0^{\beta} du 
  \int_0^1  d\lambda \lambda
\sqrt{F(\beta-u)} \Phi_t
F_{\lambda  L_t}(u)L_t \sqrt{F(\beta-u)}
\\
&+\beta \Tr
\int_0^{\beta} du 
  \int_0^1  d\lambda \lambda
\sqrt{F(\beta-u)} L_t
F_{\lambda  L_t}(u)\Phi_t \sqrt{F(\beta-u)}
\\
&-\beta \Tr
\int_0^{\beta} du 
\int_0^{\beta} du_1 
  \int_0^1  d\lambda \lambda
\sqrt{F(\beta-u)} L_t
F_{\lambda  L_t}(u-u_1)
\Phi_t
F(u_1)
L _t \sqrt{F(\beta-u)}
\end{align*}
where in the second equality we used the result of Proposition \ref{prop:derivatives}.
The well posedness of the trace written above can be discussed as in Lemma \ref{lm:S-traceclass}.
We can now compute the limit $t\to\infty$ in the relation \eqref{eq:entropy-differences}.
If return to equilibrium holds, we have that $\lim_{t\to\infty} \omega_t^\beta = \omega^{\beta K}$ and 
\[
\mathscr{S}(\omega^\beta, \omega^{\beta K} )
-
\mathscr{S}(\omega^\beta, \omega^{\beta} )
=
\lim_{t\to\infty}
\int_0^t E_{\tilde{t}}\, d\tilde{t}
\]
where $\mathscr{S}(\omega^\beta, \omega^{\beta K} ) = \mathscr{S}_{0,K}$ given in \eqref{eq:rel:entropy-trace}.
If $\mathscr{S}(\omega^\beta, \omega^{\beta K} )$  is finite and since
$\mathscr{S}(\omega^\beta, \omega^{\beta} )=0$ we have that 
\[
\lim_{t\to \infty} E_t =0 \ .
\]
On the other end, if no return to equilibrium holds, we expect that, if  $\omega^\beta_t$  converges in the limit $t\to \infty$ to a state which is not normal to the $\omega^\beta$ (after the ergodic mean), the corresponding relative entropy diverges.
In this case, it is thus expected that the entropy production converges to a constant.

\section*{Acknowledgments}
The authors want to thank 
the referees for pointing out some relevant literature, 
which was useful to clarify  the relation of our results  to existing studies.
NP is  grateful for the support of the National Group of Mathematical Physics (GNFM-INdAM).
The research conducted by NP was supported in part by the MIUR Excellence Department Project 2023-2027 awarded to the Dipartimento di Matematica of the University of Genova, CUPD33C23001110001.

\appendix
\section{}

\subsection{Quasi-equivalence}

In this appendix we prove that the states $\omega^{\beta K}$ and $\omega^\beta$ which are the quasi free states constructed out of the operators $(1+e^{-\beta(D+\mathcal{K})})^{-\frac12}$ and $(1+e^{-\beta D})^{-\frac12}$ acting on the one-particle Hilbert space $\mathcal{H}$, are quasi equivalent if $\mathcal{K}$ is sufficiently small.
To prove it we verify the Powers-St{\o}rmer criterion \cite{PowerStormer}.

\begin{proposition}\label{prop:PowerStormer}
Suppose that $\|\mathcal{K}\| \beta < 1$ and the hypotheses of Lemma \ref{lm:Cook-criterion-D+K} hold.
The Powers-St{\o}rmer criterion is satisfied by the states
$\omega^\beta$ and $\omega^{\beta K}$, hence they are normal states.
\end{proposition}

\begin{proof}
Let $A = (1+e^{-\beta (D+ \mathcal{K})})^{-1}$ and $B = (1+e^{-\beta D})^{-1}$.
We need to prove that 
$A^{\frac{1}{2}} - B^{\frac{1}{2}} $ and 
$(1-A)^{\frac{1}{2}} - (1-B)^{\frac{1}{2}}$ are Hilbert-Schmidt operators. 

We discuss in detail the finiteness of the Hilbert-Schmidt norm of the first one. Observing that 
\[
1-\frac{1}{1+e^{-\beta D}} = \frac{1}{1+e^{\beta D}} 
\]
the finiteness of the Hilbert-Schmidt norm of the second operator follows similarly.

Notice that 
\[
A^{\frac{1}{2}} - B^{\frac{1}{2}} = 
\Omega_+ \left(\frac{1}{1+e^{-\beta D}}
\right)^{\frac{1}{2}}\Omega_+^{*}
-
\left(\frac{1}{1+e^{-\beta D}}\right)^{\frac{1}{2}}
\]
where $\Omega_+$ is the M{\o }ller operator discussed above.
We decompose $\Omega_+ =\tilde{\Omega} +1$ and we obtain
\[
A^{\frac{1}{2}} - B^{\frac{1}{2}} = 
\tilde{\Omega} \left(\frac{1}{1+e^{-\beta D}}
\right)^{\frac{1}{2}}
+
\left(\frac{1}{1+e^{-\beta D}}
\right)^{\frac{1}{2}}\tilde{\Omega}^{*}
+
\tilde{\Omega} \left(\frac{1}{1+e^{-\beta D}}
\right)^{\frac{1}{2}}\tilde{\Omega}^{*}
\]
Let us recall that 
$P\mathcal{K}Q$
is Hilbert-Schmidt as we have shown in  
Lemma \ref{lm:KHSunifrom} and \ref{lm:KHS}.
It follows that $P\tilde{\Omega}Q$ is Hilbert-Schmidt too.
The Fourier transform of the integral kernel of $\tilde{\Omega}$ is bounded by $\tilde{H}(p-q)$ which is an integrable and $L^2$ function.

Furthermore, the Fourier transform of $Q (1+e^{-\beta D})^{-\frac{1}{2}}$ is an $L^1 \cap L^2$ function because it is bounded and decays rapidly for large momenta.

Hence, composing with $P+Q$ on both sides of all the operators appearing in  $A^{\frac{1}{2}} - B^{\frac{1}{2}}$,  we get that the contributions which contain at least one $Q$ are Hilbert-Schmidt. 

It remains to discuss
\begin{align*}
R &= 
P \Omega_+ P \left(\frac{1}{1+e^{-\beta D}}
\right)^{\frac{1}{2}}P\Omega_+^{*}P
-
P\left(\frac{1}{1+e^{-\beta D}}\right)^{\frac{1}{2}}P\\
&= 
\lim_{T\to\infty}
Pe^{i T (D+\mathcal{K})}P e^{-i T D}
\left(\frac{1}{1+e^{-\beta D}}
\right)^{\frac{1}{2}}
e^{i T D} Pe^{-i T (D+\mathcal{K})}P
-
P\left(\frac{1}{1+e^{-\beta D}}\right)^{\frac{1}{2}}
\\
&= 
\lim_{T\to\infty}
\int_0^T dt 
Pe^{i t (D+\mathcal{K})} Pe^{-i t D}
i[\mathcal{K}(t),\left(\frac{1}{1+e^{-\beta D}}
\right)^{\frac{1}{2}}]
e^{i t D} P e^{-i t (D+\mathcal{K})}P
\\
&= 
\lim_{T\to\infty}
\int_0^T dt 
Pe^{i t (D+\mathcal{K})} P
i[\mathcal{K},\left(\frac{1}{1+e^{-\beta D}}
\right)^{\frac{1}{2}}]
P e^{-i t (D+\mathcal{K})} P
\end{align*}
Notice that 
\[
O=P[\mathcal{K},({1+e^{-\beta D}})^{-\frac{1}{2}}] P = P[D,G]P
\]
where $G$ is a Hilbert-Schmidt operator.
To prove the last claim we compute the  Fourier transform of (the integral kernel of ) $O$, we get 
\begin{align*}
O(p,q)
&=
 \left[
 \mathcal{K}(p,q)
 \left(\frac{1}{\sqrt{1+e^{-\beta \omega(p)}}}
 -
 \frac{1}{\sqrt{1+e^{-\beta \omega(q)}}}
\right)
\right]
\\
&=
-\int_0^1 d\lambda
 \left[
 (\omega(p)-\omega(q)) 
\mathcal{K}(p,q)
 \frac{\beta}{2}\left(\frac{e^{-\beta \lambda \omega(p)}
 e^{-\beta (1-\lambda) \omega(q)}}{\sqrt{(1+e^{-\beta \lambda \omega(p)}
 e^{-\beta (1-\lambda) \omega(q)})^3}}
\right)
\right]
\end{align*}
The integral kernel of the operator $G$ is thus
\[
G(p,q) = -\int_0^1 d\lambda
 \left[
\mathcal{K}(p,q)
 \frac{\beta}{2}\left(\frac{e^{-\beta \lambda \omega(p)}
 e^{-\beta (1-\lambda) \omega(q)}}{\sqrt{(1+e^{-\beta \lambda \omega(p)}
 e^{-\beta (1-\lambda) \omega(q)})^3}}
\right)
\right]
\]
and this is the integral Kernel of an Hilbert-Schmidt operator because, 
$\mathcal{K}(p,q)$ is bounded by an $L^1\cap L^2$ function of $(p-q)$ thanks to the results of Lemma \ref{lm:Kbounded}. Furthermore, for every value of $\lambda \in [0,1]$ the function in round brackets is $L^2$ either in $p$ or $q$ and bounded in the other variable.

Observe now that $O=P[D+\mathcal{K},G]P-P[\mathcal{K},G]P$. The corresponding contribution in $R$ involving $D+\mathcal{K}$ is a total derivative in time, hence it can be integrated and the result is bounded uniformly in $T$. Having established  the validity of the Cook criterion in Lemma \ref{lm:Cook-criterion-D+K}, the integral in time of the reminder term is bounded too. More precisely, we have 
\begin{align*}
R=\lim_{T\to\infty}
\int_0^T dt 
Pe^{i t (D+\mathcal{K})} 
i
P[D,G]P
 e^{-i t (D+\mathcal{K})}
P=&
\lim_{T\to\infty}
Pe^{i T (D+\mathcal{K})} 
PGP
 e^{-i T (D+\mathcal{K})}
 P
-
PGP
+\\
&-
\lim_{T\to\infty}
\int_0^T dt 
Pe^{i t (D+\mathcal{K})} 
i[\mathcal{K},PGP]
 e^{-i t (D+\mathcal{K})}
 P
\end{align*}
$G$ is Hilbert-Schmidt, $e^{iT (D+K)}$, $P$ are bounded operator of norm $1$, since $G$ is Hilbert-Schmidt we have that the first two contributions at the right hand side are Hilbert-Schmidt. 
The last contribution can be analyzed using the Cook criterion for $D+\mathcal{K}$ and the boundedness of the exponential $e^{i t (D+\mathcal{K})}$. Since $G$ is Hilbert-Schmidt we have that also the last contribution is Hilbert-Schmidt. 
\end{proof}

\subsection{Time decays}

We recall some estimates about the decay in time of $\int e^{i\omega(p)t} f(p)d^3p$. We collect the properties used in the paper in the following lemma.

\begin{lm}\label{lm-decay}
Consider $f$ a Schwartz function on $\mathbb{R}^3$. It holds that, for $t>1$
\[
|\int e^{i\omega(p)t} f(p)d^3p| \leq \frac{C}{t^{3/2}} \left(
|f(0)|+\sum_{|\alpha|\leq 2} \| D^\alpha f\|_\infty  \right) 
+\frac{C}{t^2} \left( \sum_{|\alpha|\leq 2} \| D^{\alpha}f \|_1 \right)
\]
where $\omega(p)=\sqrt{p^2+m^2}$ with $m>0$, $\alpha$ is a multi-index and $D^\alpha$ denotes the derivatives corresponding to that multi-index.
If $f$ is the Fourier transform of $g\in C^\infty_0(\mathbb{R}^3)$, we have that 
\[
|\int e^{i\omega(p)t} \hat{g}(p)d^3p| \leq \frac{C'}{t^{3/2}} \left(
|\hat{g}(0)|+ \| \hat{g}\|_1  \right) 
\]
where the constant $C'$ depends on the domain of $g$.
\end{lm}
\begin{proof}
The oscillatory integral considered here can be treated with stationary phase methods.
The phase has a single stationary point at $p=0$. 
We consider a function $\rho \geq 0$ which is smooth, of compact support and equal to $1$ in a neighborhood of $0$ and we decompose $f=f_1+f_0$ where $f_0 = \rho f$. 

The contribution with $f_1$ can be estimated similarly to Theorem 7.7.1 in \cite{Hormander} while the contribution with $f_0$ as in  Theorem 7.7.5 of  \cite{Hormander}. 
In order to get concrete estimates for these contributions we use a strategy similar to the one that can be found in the 
Lemma A.1 in \cite{MedaPinamonti}, see also Appendix A in \cite{BrosBuchholz}.

In particular, since $\omega(p)$ is smooth and has no stationary points on the support of $f_1$ for the contribution with $f_1$ we have by partial integration
\begin{align*}
    t^2\int e^{i\omega(p)t} f_1(p)d^3p &= 
\int e^{i\omega(p)t} 
\left(\frac{(m^2 - 2 p^2) }{p^4} f_1(p) 
-
  \frac{(m^2 + 4 p^2)}{p^3} f_1'(p) - 
     \frac{(m^2 + p^2)}{p^2} f_1''(p)
 \right) d^3 p
\end{align*}
where $f'$ denotes the derivative with respect to $|p|$. Hence, noticing that $f_1$ is supported away from $p=0$,  for a suitable constant $C$ we obtain 
\begin{align*}
|t^2\int e^{i\omega(p)t} f_1(p)d^3p|  
    &\leq C\sum_{|\alpha|\leq 2} \int |D^\alpha f_1(p)| d^3p.
\end{align*}
To analyze the contribution of $f_0$ we use spherical coordinates $p=(|p|,\Omega)$ and we
change variable of integration $|p|= \sqrt{(y+2m)y}$ and then $s=ty$ With this choice
\[
I=\int e^{i \omega(p) t} f_0(p) d^3 p =
e^{i2m t}\int e^{i y t} f(\sqrt{y(y+2m)},\Omega)\sqrt{(y+2m)}(2y +2m)
\sqrt{y} dy d\Omega
\]
denoting by $g(y,\Omega)=f_0(\sqrt{y(y+2m)},\Omega)\sqrt{(y+2m)}(2y +2m)$ and separating the leading contribution at large time we have that
\begin{align*}
 I 
&=
e^{i2m t}\int e^{i y t} 
g(0)
\sqrt{y} dy d\Omega
+
e^{i2m t}\int e^{i y t} 
\left(g(y,\Omega)-g(0)\right)
\sqrt{y} dy d\Omega
\\
&= 
\lim_{\epsilon \to 0}
\frac{e^{i2m t}}{t^{\frac{3}{2}}}\int e^{(i -\epsilon) s } g(0)
\sqrt{s} ds d\Omega + R
\\
&= 
\lim_{\epsilon \to 0}
\frac{2 \pi^{\frac{3}{2}}}{ (-i +\epsilon)^{\frac{3}{2}}}
\frac{g(0)}{t^{\frac{3}{2}}} e^{i2m t}
+ R 
\end{align*}
where we have inserted an $\epsilon$ regulator which is eventually removed, and we have computed the integral of $e^{(i-\epsilon) s}\sqrt{s}$.
The first term at the right hand side is the leading order contribution
and it is bounded by $C f_0(0)$,
the reminder $R$ decays more rapidly. 
For our purpose, we prove that the reminder decays for large time in the same way.
Integrating by part in the reminder $R$ we have
\begin{align*}
R &= 
\frac{e^{i2m t}}{t^{\frac{3}{2}}}\int e^{i s }  
\left(g(\frac{s}{t},\Omega) -g(0)\right)
\sqrt{s} ds d\Omega
= 
-\frac{e^{i2m t}}{t^{\frac{3}{2}}}\int \left(e^{i s }-1\right)  
\partial_s^2\left(\left(g(\frac{s}{t},\Omega) -g(0)\right)
\sqrt{s}\right) ds d\Omega
\end{align*}

Introducing the function 
\[
h(x) = 
x^{\frac{3}{2}} \frac{\partial^2}{\partial x^2}
\left(\left(g(x,\Omega) -g(0)\right)
\sqrt{x} \right)
\]
we have that 
\[
R = 
-\frac{e^{i2m t}}{t^{\frac{3}{2}}}\int 
\left(e^{i s } -1\right)
\frac{1}{s^{\frac{3}{2}}}
h(\frac{s}{t}) ds d\Omega
\]
$h$ is bounded while the remaining contribution in the integral is integrable.
We can thus estimate the reminder \[
|R| \leq \frac{C}{t^{\frac32}} \|h\|_\infty.
\]
Rewriting $h$ in terms of $f_0$ and $p$ we get
\[
h(p) =  m^{3/2}\frac{f_0(0)}{\sqrt{2}}
+
2 \frac{\omega(p)}{(m+\omega)^{\frac{3}{2}}}   
\left((2p^2-m^2) f_0(p)
+ (4p^2+m^2) |p|f_0'(p)
+  \omega(p)^2 p^{2} f_0''(p)
    \right)
\]
hence, in view of the fact that the support of $\rho$ and thus of $f_0$ is compact, 
\[
\|h\|_\infty \leq  C\left(|f(0)|+\sum_{|\alpha|\leq 2} \| D^\alpha f\|_\infty
\right).
\]
Combining this with the estimate for the contribution $f_1$, we get the first estimate of this Lemma.
To get the second notice that
\[
\|D^{\alpha}\hat{g}\|_\infty \leq \|x^{\alpha} g\|_1 \leq C \|g\|_\infty \leq C\|\hat{g}\|_1
\]
where we used the fact that $g$ has compact support. With this observation the first statement implies the second.
\end{proof}

\subsection{Implementability at finite temperature}
We showed that the derivation induced by $\mathcal K$
can be implemented by an unbounded selfadjoint operator in the vacuum representation by proving that $P\mathcal K Q$ is Hilbert-Schmidt. For gauge-invariant quasifree states, Lundberg \cite{Lundberg} found the following criterion, slightly corrected.
\begin{proposition}
    Let $\omega_A$ denote the gauge invariant quasifree state with 2-point function
    \be
    \omega_A(\psi(f)\psi(g)^*)=\langle f,A g\rangle
    \ee
    with $0\le A\le 1$. Let $T$ be bounded and self-adjoint and such that $\sqrt{A}T\sqrt{1-A}$ is Hilbert-Schmidt. Then there exists a densely defined selfadjoint operator $\mathcal T$ in the GNS-representation $(\mathfrak H,\pi,\Omega)$ induced by $\omega_A$ such that
    \be
    e^{it\mathcal T}\pi(\psi(f))e^{-it\mathcal T}=\pi_A(\psi(e^{itT}f))\ ,
    \ee
    and $e^{it\mathcal T}$ is in the weak closure of the CAR algebra in the representation $\pi$.
\end{proposition}
We show that Lundberg's criterion is satisfied by $\mathcal K$
for the quasifree state characterized by $A=(1+e^{-\beta D})^{-1}$.
As before we use the decomposition 
\be
\sqrt A \mathcal K\sqrt{1-A}=P\sqrt A \mathcal K\sqrt{1-A}P+P\sqrt A \mathcal K\sqrt{1-A}Q+Q\sqrt A \mathcal K\sqrt{1-A}P+Q\sqrt A \mathcal K\sqrt{1-A}Q\ .
\ee
$A$ commutes with $P$ and $Q$. We already know that $P\mathcal K Q$ and $Q\mathcal K P$ are Hilbert Schmidt. Moreover, $Q\sqrt A$ and $P\sqrt{1-A}$ are exponentially decaying matrix-valued functions of momentum. Since the integral kernel of $\mathcal K$  satisfies the estimate
\be
|\mathcal K(p,q)|\le h(p-q) 
\ee
with an $\mathrm L^2$-function $h$, we conclude that also the diagonal terms $P\sqrt A \mathcal K\sqrt{1-A}P$ and $Q\sqrt A \mathcal K\sqrt{1-A}Q$ are Hilbert-Schmidt.

A similar argument shows that also the automorphism induced by the cocycle $U_t\doteq e^{it(D+\mathcal K)}e^{-itD}$ can be implemented by unitaries in the weak closure. For this purpose we use the purification (see, e. \cite{PowerStormer}) of a gauge invariant quasifree state $\omega_A$ on $\mathrm{CAR}(\mathcal H)$ to a pure quasifree state $\omega_{E_A}$ on $\mathrm{CAR}(\mathcal H\oplus\overline{\mathcal H})$ with the projection
\be
E_A\doteq \left(\begin{array}{cc}
    A & \sqrt A\sqrt{1-A}  \\
   \sqrt A\sqrt{1-A}   & 1-A
\end{array}\right)\ .
\ee
The unitary 
\be
\mathcal U\doteq
\left(\begin{array}{cc}
     U&0  \\
     0&1 
\end{array}\right)
\ee
induces an automorphism of the subalgebra associated to $\mathcal H$ and acts trivially on the subalgebra associated to $\overline{\mathcal H}$. It can be implemented if $[E_A,\mathcal U]$ is Hilbert-Schmidt. This condition is satisfied if $\sqrt A(U-1)\sqrt{1-A}$ and $\sqrt{1-A} (U-1)\sqrt{A}$ are Hilbert-Schmidt. In our situation, this follows by essentially the same arguments as above. Moreover, the implementing unitary commutes with $\pi(\mathrm{CAR}(\overline{\mathcal H}))$ and also with the unitary $\Gamma\doteq e^{i\pi Q}$, where $Q$ is the charge operator. Hence, after a Klein transformation induced by $\Gamma$, it commutes with the commutant of $\pi(\mathrm{CAR}(\mathcal H))$, thus it is in the weak closure.

Actually, $\mathcal U$ admits an expansion as a power series of time ordered products of the operator valued function $t\mapsto e^{itH}Ke^{-itH}$, where $H$ is the modular Hamiltonian associated to $\omega_A$. For finite rank operators $\mathcal{K}_1,\dots,\mathcal{K}_n$ we have the formula for the truncated $n$-point function
of the operators $K_i=\sum_j \psi(f_{ij})\psi(g_{ij})^*$
\be\label{eq:formula-omega_truncated}
\omega_A^c(K_1,\dots,K_n)=(-1)^n\sum_{\sigma\in P^1_n}\Tr A_{\sigma_n1}\mathcal{K}_1A_{1\sigma_2}\dots A_{\sigma_{n-1}\sigma_n}\mathcal{K}_n
\ee
with $A_{ij}\doteq A$ for $i<j$, $A_{ij}=(A-1)$ for $i>j$ and $P^1_n$ the group of permutations $\sigma$ of $\{1,\dots,n\}$ with $\sigma_1=1$. To generalize the formula to operators $K_i$ satisfying the Lundberg condition, one may rewrite the trace for $n>1$ as
\[
\Tr\prod_1^n\sqrt{A_{\sigma_{i-1}\sigma_i}}\mathcal{K}_i\sqrt{A_{\sigma_i\sigma_{i+1}}}\]
with $\sigma_0\doteq n$ and $\sigma_{n+1}=1$.  Since for at least two values of $i$ the orders of the pairs $\sigma_{i-1}\sigma_i$ and $\sigma_{i}\sigma_{i+1}$ differ, the corresponding factors are Hilbert-Schmidt. Since the other factors are bounded operators, the full product is trace class, and the truncated $n$-point function is well defined.

For the expansion of the expectation value of $\mathcal U(t)$, $t>0$, we get
\be
\omega_A(\mathcal{U}(t))=\exp \sum_{n=1}^\infty\int_{t>t_1>\dots t_n>0}\omega_A^c(K(t_1),\dots ,K(t_n))\
\ee
with
\be
\left|\int_{t>t_1>\dots t_n>0}dt_1\dots dt_n\omega_A^c(K(t_1),\dots ,K(t_n))\right|\le \frac{1}{n}||\sqrt{A}\mathcal K\sqrt{1-A}||_{HS}^2 ||\mathcal{K}||^{n-2}t^n 
\ee
for $n>1$. The term of 1st order has to be renormalized, but the terms of higher order can be summed provided $t||K||<1$.

The integrand can be analytically extended in the time variables to the region 
\be
\{(z_1,\dots,z_n)\in\mathbb{C}|\beta>\Im z_1>\dots\Im z_n>0\}.
\ee
Once the boundedness of the function is proven, iterated application of the three-line theorem yields for $\beta>u_1>\dots u_n>0$ the estimate
\begin{equation}
\label{eq:estimate-truncated}
|\omega_A^c(K(iu_1)\dots K(iu_n))|\le (n-1)!||\sqrt{A}\mathcal K\sqrt{1-A}||_{HS}^2 ||\mathcal{K}||^{n-2}\ .    
\end{equation}
\subsection{Power expansion of the perturbed Fermi factor}\label{se:derivatives}
\begin{proposition}\label{prop:derivatives}
    Let $D\doteq D_0+\lambda K$ with $D_0$ selfadjoint and $K$ selfadjoint and bounded. Let
    \begin{equation}
        A(u)\doteq(e^{uD}+e^{(u-1)D})^{-1}
    \end{equation}
    for $u\in[0,1)$ and $A(u)=(-1)^{[u]}A(u-[u])$ for $u\in\mathbb{R}$, where $[u]$ denotes the largest integer $\le u$.

    Then
    \begin{equation}
        \frac{d^n}{d\lambda^n}A(u)=(-1)^n n!\int_{0<x_i<1}dx_1\dots dx_n A(u-\sum x_i)\prod_{i=1}^nKA(x_i)\ .
    \end{equation}
    
\end{proposition}
\begin{proof}
    Let $n=1$. We have
    \begin{equation}
        \frac{d}{d\lambda}A(u)=-A(u)\frac{d}{d\lambda}(e^{uD}+e^{(u-1)D})A(u)\ .
    \end{equation}
    Using the formula
    \begin{equation}
        \frac{d}{d\lambda}e^{uD}=\int_0^1 dv e^{vuD}uKe^{(1-v)uD}
    \end{equation}
    and the corresponding formula for $e^{(u-1)D}$,
    we obtain
    \begin{equation}
     \frac{d}{d\lambda}A(u)=-A(u)\int_0^1 dv(e^{vuD}uKe^{(1-v)uD}+e^{{v(u-1)D}}(u-1)Ke^{(1-v)(u-1)D})A(u) \ .  
    \end{equation}
    We now use the fact that for $v\in[0,1]$
    \begin{equation}
        A(u)e^{vuD}=A(u-uv)
    \end{equation}
    holds. By this formula all the exponentials of $D$ can be absorbed by modifying the arguments of $A$, and we obtain
    \begin{equation}
        \frac{d}{d\lambda}A(u)=-\int_0^1
dv(A(u-uv)uKA(uv)+A(u-(u-1)v)(u-1)KA(1+(u-1)v)    \end{equation}
Note that the arguments of $A$ in the first term vary between $0$ and $u$ and in the second term between $u$ and $1$.
We introduce new variables: In the first term we set $x\doteq uv$ with $dx=udv$, in the second term $x\doteq 1+(u-1)v$ with $dx=(u-1)dv$. We get
\begin{equation}
 \frac{d}{d\lambda}A(u)=-\int_0^u dx A(u-x)KA(x)+\int_u^1 dx A(u+1-x)KA(x)\ .   
\end{equation}
By the definition of $A$ for arguments outside of $[0,1)$ we have $A(u+1-x)=-A(u-x)$, hence get the simple formula
\begin{equation}
    \frac{d}{d\lambda}A(u)=-\int_0^1 dx A(u-x)KA(x)\ .
    \end{equation}
We now proceed by induction and assume the formula for the $n$th derivative. Then
\begin{equation}
    \begin{split}
        &\frac{d^{n+1}}{d\lambda^{n+1}}A(u)=(-1)^nn!\int_{0<x_i<1}dx_1\dots dx_n \frac{d}{d\lambda}\bigl(A(u-\sum x_i)\prod_{i=1}^n(KA(x_i))\bigr)\\
     = &(-1)^{n+1}n!\int_{0<x_i<1}dx_1\dots dx_n \int_0^1 dx \left(A(u-\sum x_i-x)KA(x)\prod_{i=1}^n(KA(x_i))\right.\\
     &+\left.\sum_{k=1}^n A(u-\sum x_i)K\prod_{i=1}^{k-1}(A(x_i)K)A(x_k-x)KA(x)\prod_{i=k+1}^n(KA(x_i))\right)  
    \end{split}
\end{equation}
In a last step we introduce new variables: in the first term we set $y_1\doteq x, y_i\doteq x_{i-1}, i>1$ and in the term labeled by $k$, we set $y_i\doteq x_i$, $i<k$, $y_k\doteq x_k-x$, $y_{k+1}\doteq x$ and $y_i\doteq x_{i-1}$, $i>(k+1)$. 
In the term labeled by $k$, the integration domain for $y_k$ is $[-x,1-x]$. For $y_k<0$ we replace it by $y_k+1$. This changes the sign for $A(y_k)$ and for $A(u-\sum y_k)$, thus all the $(n+1)$ terms give the same contribution, and we obtain the proposed formula or the $(n+1)$th derivative.    
\end{proof}

{\bf Data availability:} Data sharing not applicable to this article as no datasets were generated or analyzed during
the current study. 

{\bf Conflicts of interest statement:} All authors have no conflicts of interest.


\begin{thebibliography}{99}
{\small 



\bibitem{AB2} V.~Abram and R.~Brunetti, 
\emph{``Dynamical C*-algebras and quadratic perturbations for fermions,''} in preparation


\bibitem{Altherr}
T.~Altherr {\it ``Infrared problem in $g\phi^4$ theory at finite temperature,''}
Phys.\ Lett.\ {\bf B 238} (1990), 360


\bibitem{ArakiSelfDual}
H.~Araki,
\emph{``On quasifree states of CAR and Bogoliubov automorphisms,''} 
Res. Inst. Math. Sci. {\bf 6} (1970), 385–442.

\bibitem{Araki-KMS}
H.~Araki,
\emph{``Relative Hamiltonian for faithful normal states of a von Neumann algebra,''} 
Publ.\ RIMS, Kyoto Univ.\ {\bf 9}(1), (1973) 165-209

\bibitem{Araki}
H.~Araki,
\emph{``Relative Entropy of States of von Neumann Algebras,''}
Publ.\ RIMS, Kyoto Univ.\ {\bf 11}, (1976) 809-833

\bibitem{Araki2}
H.~Araki,
\emph{``Relative Entropy of States of von Neumann Algebras II,''}
Publ.\ RIMS, Kyoto Univ.\ {\bf 13}, (1977) 173-192


\bibitem{ArakiBook}
H.~Araki,
\emph{``Mathematical Theory of Quantum Fields,''}
Oxford University Press 1999


\bibitem{BaFreRe} D. Bahns, K. Fredenhagen and K. Rejzner,
\emph{``Local nets of von Neumann algebras in the Sine-Gordon model,''}
Commun. Math. Phys. 383 (2021) 1, 1-33

\bibitem{BaRe} D.~Bahns and K.~Rejzner, 
\emph{``The quantum Sine-Gordon model in perturbative AQFT,''}
Commun.\ Math.\ Phys.\  {\bf 357} (2018) 421-446

\bibitem{Birman}
M. S. Birman, 
\emph{``On the spectrum of singular boundary-value problems,''} 
Mat. Sb. \textbf{55} (1961)2, 125-174.



\bibitem{Muller}
L.~Bollmann, P.~Müller,
\emph{``Enhanced area law in the Widom–Sobolev formula for the free Dirac operator in arbitrary dimension,''}
Pure Appl. Anal
\textbf{7} (2025), 595--613

\bibitem{Joao}
J. Braga de Góes Vasconcellos, N. Drago and N. Pinamonti,
\emph{``Equilibrium States in Thermal Field Theory and in Algebraic Quantum Field Theory,''} 
Ann. Henri Poincaré {\bf 21} (2020), 1--43

\bibitem{BKR}
O.~Bratteli, A.~Kishimoto,  D.W.~Robinson, 
\emph{``Stability properties and the KMS condition,''} 
Commun.\ Math.\ Phys.\ {\bf 61} (1978), 209-238


\bibitem{BratteliRobinson}
O.~Bratteli,  D.W.~Robinson, 
\emph{``Operator algebras and quantum statistical mechanics 2,''} 
Springer, Berlin 1997


\bibitem{BrosBuchholz-RelativisticKMS}
J. Bros, D. Buchholz,
\emph{``Towards a relativistic KMS-condition,''}
Nuc. Phys. \textbf{B429} (1994)2,
291-318


\bibitem{BrosBuchholz}
J. Bros, D. Buchholz,
\emph{``Asymptotic dynamics of thermal quantum fields,''}
Nuc. Phys. \textbf{B627} (2002)1-2,
289-310

\bibitem{BDF}
 R.~Brunetti, M.~D\"utsch and K.~Fredenhagen,
\emph{``Perturbative algebraic quantum field theory and the 
renormalization groups'',} 
Adv.\ Theor.\ Math.\ Phys.\  {\bf 13} (2009) 1541-1599

\bibitem{BDFR-Fermi}
R.~Brunetti, M.~D\"utsch, K.~Fredenhagen and K.~Rejzner,
\emph{``C*-algebraic approach to interacting quantum field theory: inclusion of Fermi fields'',}
Letters in Mathematical Physics \textbf{112} (2022)5,  101

\bibitem{BDFR-MWI}
R.~Brunetti, M.~D\"utsch, K.~Fredenhagen and K.~Rejzner,
\emph{``The unitary master Ward identity: time slice axiom, Noether's theorem and anomalies'',}
Ann. Henri Poincar\'e  \textbf{24} (2023) 2, 469-539

\bibitem{BDFR-Gauge}
R.~Brunetti, M.~D\"utsch, K.~Fredenhagen and K.~Rejzner,
\emph{``Haag-Kastler Nets For Gauge Theories'',}
work in progress.

\bibitem{BFP}
R. Brunetti, K. Fredenhagen and N. Pinamonti, 
\emph{``Algebraic Approach to Bose–Einstein Condensation in Relativistic Quantum Field Theory: Spontaneous Symmetry Breaking and the Goldstone Theorem,''} 
Ann. Henri Poincar\'e \textbf{22} (2021), 951–1000



\bibitem{BF19}
 D.~Buchholz and K.~Fredenhagen,
 \emph{``A C*-algebraic approach
to interacting quantum field theories'', }
Commun. Math. Phys. \textbf{377} (2020), 947--969

\bibitem{BF21}
D.~Buchholz and K.~Fredenhagen, 
\emph{``Dynamical C*-algebras and kinetic perturbations,''}
Annales Henri Poincare \textbf{22} (2021) 3, 1001-1033


\bibitem{Cadamuro}
D. Cadamuro, M.B. Fröb and C. Minz, \emph{``Hamiltonian for Fermions of Small Mass,''} 
Ann. Henri Poincaré (2024). 
DOI 10.1007/s00023-024-01508-0


\bibitem{Calzetta-Hu}
E.~Calzetta and B.~L.~Hu, 
{\it ``Nonequilibrium quantum fields: Closed-time-path effective action, Wigner function, and Boltzmann equation, ''} Phys. Rev. {\bf D 37} (1988), 2878-2900

\bibitem{Casini} 
H. Casini, S. Grillo, and D. Pontello,
\emph{``Relative entropy for coherent states from Araki formula,''} 
Phys. Rev. D \textbf{99} (2019) 125020

\bibitem{Casini2008} 
H. Casini and M. Huerta,
\emph{``Analytic results on the geometric entropy for free fields,''}
J. Stat. Mech. (2008) P01012

\bibitem{CiolliLongoRuzzi}
F. Ciolli, R.  Longo and G.  Ruzzi
\emph{``The Information in a Wave,''}
Commun. Math. Phys. 
   \textbf{379} (2019), 979–1000

\bibitem{cojuhari}
P. A. Cojuhari, 
\emph{``On the finiteness of the discrete spectrum of the Dirac operator,''}
Rep. Math. Phys.  
\textbf{57} (2006)3, 334-341


\bibitem{AdvancedCombinatorics}
L. Comtet, 
\emph{``Advances Combinatorics,''}
The Art of Finite and Infinite Expansions 
D. Reidel Publishing Company, Dordrecht, Holland 1974



\bibitem{DeckertMerkel}
DA.~Deckert, F.~Merkl, 
\emph{``External Field QED on Cauchy Surfaces for Varying Electromagnetic Fields,''}
Commun. Math. Phys. \textbf{345} (2016),  973-1017



\bibitem{Dimock}
J.~Dimock,
\emph{``Dirac Quantum Fields on a Manifold,''}
Trans. Am. Math. Soc. \textbf{269} (1982), 133-147


\bibitem{Drago}
N. Drago, 
\emph{``Thermal State with Quadratic Interaction,''}
Ann. Henri Poincaré {\bf 20} (2019), 905–927 


\bibitem{DFP}
N. Drago, F. Faldino and N. Pinamonti,
\emph{``On the Stability of KMS States in Perturbative Algebraic Quantum Field Theories,''}
Commun. Math. Phys. \textbf{357} (2018), 267–293

\bibitem{DFP2}
N. Drago, F. Faldino and N. Pinamonti,
\emph{``Relative Entropy and Entropy Production for Equilibrium States in pAQFT,''} 
Ann. Henri Poincaré \textbf{19} (2018), 3289–3319

\bibitem{Duetsch}
M. D\"utsch,
\emph{``From Classical Field Theory to Perturbative Quantum Field Theory.''} 
Springer, 2019

\bibitem{DF}
M. D\"utsch and K. Fredenhagen, 
\emph{``Causal perturbation theory in terms of
retarded products, and a proof of the Action Ward Identity,''} 
Rev. Math. Phys. \textbf{16} (2004), 1291–1348

\bibitem{Finster}
F. Finster, R. H. Jonsson,  M. Lottner, A. Much and S. Murro,
\emph{``Notions of Fermionic Entropies of a Causal Fermion System''},
Math. Phys. Anal. Geom. \textbf{28} (2025), 7 


\bibitem{FinsterLottinerSobolev}
F.~Finster, M.~Lottiner and A.~V.~Sobolev,
\emph{``The fermionic entanglement entropy and area law for the relativistic Dirac vacuum state,''}
Adv.\ Theor.\ Math.\ Phys.\  \textbf{28} (2024), 1933-1985



\bibitem{FinsterMurroRoken}
F.~Finster,  S.~Murro and C.~R\"oken,
\emph{``The fermionic projector in a time-dependent external potential: mass oscillation property and Hadamard states,''} 
J. Math. Phys. \textbf{57} (2016), 072303



\bibitem{Fredenhagen}
K.~Fredenhagen,
\emph{``Implementation of Automorphisms and Derivatives of the CAR-Algebra,''}
Commun. Math. Phys. \textbf{52} (1977), 255-266

\bibitem{FredenhagenQED1} 
K.~Fredenhagen,
\emph{``Quantum electrodynamics with one degree of freedom for the photon field,''}
Ann. Inst. H. Poincare Phys. Theor. \textbf{27} (1977), 319-334

\bibitem{FL}
K.~Fredenhagen and F.~Lindner,
\emph{``Construction of KMS States in Perturbative QFT and Renormalized Hamiltonian Dynamics,''}
Commun. Math. Phys. \textbf{332} (2014) 3, 895-932
[erratum: Commun. Math. Phys. \textbf{347} (2016) 2, 655-656]

\bibitem{Frob}
M. B. Fröb,
\emph{``Relating the modular Hamiltonian to two-point functions,''}
arXiv:2501.09669 (2025)


\bibitem{Frob1}
M. B. Fröb, L. Sangaletti, 
\emph{``Petz–Rényi relative entropy in QFT from modular theory,''} 
Lett Math Phys {\bf 115}, (2025) 30. 


\bibitem{Galanda1}
S. Galanda, 
\emph{``Perturbative construction of equilibrium states for interacting fermionic field theories. Semiclassical Maxwell equation and the Debye screening length,''}
Comm. Contemp. Mathematics, online ready, https://doi.org/10.1142/S0219199725500191


\bibitem{Galanda}
S.~Galanda, A.~Much, R.~Verch,
\emph{``Relative Entropy of Fermion Excitation States on the CAR Algebra,''}
Math. Phys. Anal. Geom. \textbf{26} (2023)3, 21



\bibitem{HHW}
R.~Haag, N.M.~Hugenholtz, M.~Winnink, 
\emph{``On the equilibrium states in quantum statistical mechanics,''}
Commun.\ Math.\ Phys.\ {\bf 5} (1967), 215-236

\bibitem{HKT}
R.~Haag, D.~Kastler, E.~B.~Trych-Pohlmeyer, 
\emph{``Stability and Equilibrium States,''}
Commun.\ Math.\ Phys.\ {\bf 38} (1974), 173-193

\bibitem{HaagTrych} 
R. Haag and E. Trych-Pohlmeyer, \emph{``Stability Properties and Equilibrium States,''} Commun. Math. Phys. \textbf{56}(1977), 213


\bibitem{Helling}
R.~Helling, H.~Leschke, W.~Spitzer, 
\emph{``A special case of a conjecture by Widom with implications to
fermionic entanglement entropy'',} Int. Math. Res. Not. IMRN (2011), no. 7,
1451–1482 


\bibitem{Hollands}
S. Hollands,
\emph{``Relative entropy for coherent states in chiral CFT,''} 
Lett. Math. Phys. \textbf{110} (2020), 713–733


\bibitem{Hormander}
L. H\"ormander,
\emph{``The Analysis of Linear Partial Differential Operators I,''}
Springer, Berlin 2003

\bibitem{Stellarators}
L-M. Imbert-G\'erard, E. J. Paul and A. M. Wright,
\emph{``An introduction to Stellarators. From Magnetic Fields to Symmetries and Optimization,''}
SIAM, 2025




\bibitem{JP01}
V.~Jak{\v s}i\'c, C.-A.~Pillet,
\emph{``On entropy production in quantum statistical mechanic,''}
Commun.\ Math.\ Phys.\ \textbf{217} (2001), 285

\bibitem{JP02}
V.~Jak{\v s}i\'c, C.-A.~Pillet,
\emph{``Non-Equilibrium Steady States of Finite Quantum Systems Coupled to Thermal Reservoirs,''}
Commun.\ Math.\ Phys.\ \textbf{226} (2002), 131-162

\bibitem{JP03} 
V.~Jak{\v s}i\'c, C.-A.~Pillet, \emph{``Mathematical Theory of Non-Equilibrium Quantum Statistical Mechanics'',} J. Stat. Phys. \textbf{108} (2002) 516, 787-829


\bibitem{Kapusta}
J.~I.~Kapusta and C.~Gale 
\emph{``Finite-Temperature Field Theory - Principles and Applications''}
Cambdridge University Press (2006)

\bibitem{Kato51}
T.~Kato
\emph{``Fundamental properties of Hamiltonian operators of Schr\"odinger type"},
\textbf{70}, Trans. Amer. Math. Soc. (1951) 195-211

\bibitem{Kato}
T. Kato,
\emph{``Perturbation Theory for Linear Operators,''
}
Springer Berlin, Heidelberg 1995


\bibitem{Klaus}
M. Klaus,
\emph{``On the point spectrum of Dirac operators,''}
Helv. Phys. Acta \textbf{53} (1980)
3,  453-462

\bibitem{KL}
S. Kullback, R.A. Leibler,
\emph{``On information and sufficiency,''}
Ann. Math. Stat. \textbf{22} (1951) 1 79-86

\bibitem{Landsman}
N.~P.~Landsman, C.~G.~van Weert,
\emph{``Real and imaginary time field theory at finite temperature and density,''} 
Phys.\ Rep.\  {\bf 145} (1987), 141


\bibitem{Lang}
E.~Langmann and J.~Mickelsson,
\emph{``Scattering matrix in external field problems,''}
J. Math. Phys. \textbf{37} (1996), 3933-3953

\bibitem{Laz}
D.~Lazarovici,
\emph{``Time Evolution in the external field problem of Quantum Electrodynamics,''}
[arXiv:1310.1778 [math-ph]]


\bibitem{LeBellac} 
M.~Le Bellac, 
\emph{``Thermal Field Theory,''}
Cambridge University Press, Cambridge (2000)


\bibitem{Lindner}
F.~Lindner,
\emph{``Perturbative Algebraic Quantum Field Theory at Finite Temperature,''}
PhD thesis, University of Hamburg (2013) 



\bibitem{Longo} 
R. Longo,
\emph{``Entropy of coherent excitations,''}
Lett. Math. Phys. \textbf{109} (2019), 2587–2600. 

\bibitem{Lundberg}
L.~E.~Lundberg,
``Quasifree Second Quantization,''
Commun. Math. Phys. \textbf{50} (1976), 103-112

\bibitem{MedaPinamonti}
P. Meda, N.  Pinamonti, 
\emph{``Linear Stability of Semiclassical Theories of Gravity,''} Ann. Henri Poincaré \textbf{24} (2023),1211–1243



\bibitem{Mic}
J.~Mickelsson,
\emph{``The phase of the scattering operator from the geometry of certain infinite-dimensional groups,''}
Lett. Math. Phys. \textbf{104} (2014), 1189-1199


\bibitem{Nicolas}
J.-P. Nicolas,
\emph{``Scattering of linear Dirac fields by a spherically symmetric black hole,''} 
Ann. Inst. H. Poincar\'e Phys. Th\'eor. \textbf{62} (1995), 145-179



\bibitem{OhyaPetz}
M. Ohya, D. Petz
\emph{``Quantum Entropy and Its Use,''}
Springer Berlin Heidelberg 1993




\bibitem{Oj1}
I.~Ojima, 
\emph{``Entropy production and non-equilibrium stationarity in quantum dynamical systems: physical meaning of van Hove limit,''}
J.\ Stat.\ Phys. {\bf 56}  (1989), 203

\bibitem{Oj2}
I.~Ojima, 
\emph{``Entropy production and non-equilibrium stationarity in quantum dynamical systems,''} 
In: Proceedings of international workshop on quantum aspects of optical communications. Lecture
Notes in Physics 378, 164. Berlin: Springer-Verlag, 1991


\bibitem{Oj0}
I.~Ojima, 
H.~Hasegawa, 
M.~Ichiyanagi,
\emph{``Entropy production and its positivity in nonlinear response theory of quantum dynamical systems,''} 
J. Stat. Phys. {\bf 50} (1988), 633

\bibitem{Palmer}
J. Palmer, \emph{``Scattering automorphisms of the Dirac field,''} 
J. Math. Ana. and Appl. \textbf{64} (1978) 189-215

\bibitem{Petz}
D. Petz,
\emph{``A Variational Expression for the Relative Entropy,''}
Commun. Math. Phys. \textbf{114} (1988), 345-349

\bibitem{PfeifferSpitzer}
P.~Pfeiffer, W.~Spitzer,  
\emph{``Entanglement Entropy of Ground States of the Three-Dimensional Ideal Fermi Gas in a Magnetic Field,''} 
Ann. Henri Poincaré 
\textbf{25}  (2024), 3649–3698

\bibitem{PowerStormer}
R.~T.~Powers, A.~St{\o}rmer, 
\emph{``Free States of the Canonical Anticommutation Relations,''}
Commun. Math. Phys. \textbf{16} (1970), 1-33

\bibitem{ProkopecA}  
T.~Prokopec, M.~G.~Schmidt, and S.~Weinstock, 
\emph{``Transport equations for chiral fermions to order h bar and electroweak baryogenesis. Part I,''} 
Annals Phys. {\bf 314} (2004), 208-265  

\bibitem{ProkopecB}  
T.~Prokopec, M.~G.~Schmidt, and S.~Weinstock, 
\emph{``Transport equations for chiral fermions to order h-bar and electroweak baryogenesis. Part II,''}
Annals Phys. {\bf 314} (2004), 267-320  

\bibitem{Prosser}
R.~T. Prosser, \emph{``Relativistic Potential Scattering''} J. Math. Phys. 4 (1963), 1048–1054

\bibitem{Rejzner}
K. Rejzner,
\emph{``Perturbative Algebraic Quantum Field Theory: An Introduction for Mathematicians.''} Mathematical Physics Studies. New York: Springer, 2016


\bibitem{Ruelle}
D. Ruelle,
\emph{``Natural nonequilibrium states in quantum statistical mechanics,''}
J. Statistical Phys. \textbf{98} (2000), 57-75


\bibitem{Ruelle2} 
D. Ruelle, \emph{``Entropy production in quantum spin systems,''} Commun. Math. Phys. \textbf{224} (2001), 3

\bibitem{Ruelle3} 
D. Ruelle, \emph{``How should one define entropy production for nonequilibrium quantum spin systems?,''} Rev. Math. Phys. $\mathbf{14}$,  (2002)7-8, 701-707

\bibitem{Ruijsenaars77}
S. N. M. Ruijsenaars, \emph{``Charged particles in external fields; I. Classical theory,''} J. Math. Phys. \textbf{18} (1977), 720-737


\bibitem{Scharf}
G.~Scharf,
\emph{``Finite Quantum Electrodynamics: The causal approach''}
3rd edition,
Dover Publications, Mineola, New York 2014

\bibitem{SW}
G.~Scharf and W.F.~Wreszinski,
\emph{``The causal phase in Quantum Electrodynamics''}
Nuovo Cimento \textbf{93 A} (1986), 1

\bibitem{SS}
D. Shale and W. F. Stinespring,
\emph{``States of the Clifford Algebra'',}
Annals of Mathematics,
Second Series, \textbf{80} (1964), 365-381 

\bibitem{Seipp}
H.~P.~Seipp,
\emph{``On the S operator for the external field problem of QED,''}
Helv. Phys. Acta \textbf{55} (1982), 1-28

\bibitem{Si}
B. Simon,
\emph{``Notes on infinite determinants of Hilbert space operators'',}
Advances in Mathematics \textbf{24} (1977), 244-273

\bibitem{Spohn} 
H. Spohn, \emph{``Entropy production for quantum dynamical semigroups,''} J. Math. Phys. \textbf{19} (1978), 227

\bibitem{Stora}
R. Stora, \emph{``Renormalized perturbation theory: A missing chapter'',} 
Int. J. Geom. Meth. Mod. Phys. \textbf{5} (2008), 1345

\bibitem{Thaller}
B. Thaller,
\emph{``The Dirac Equation,''}
Springer Berlin Heidelberg 1992


\bibitem{Tokamaks}
J. Wesson,
\emph{``Tokamaks,''}
International Series of Monographs in Physics, fourth ed., Oxford University Press 2011

}
\end{thebibliography}
\end{document}